\documentclass[10pt,twocolumn,twoside]{IEEEtranTCOM}

\ifCLASSINFOpdf
\else
\fi
\usepackage[linesnumbered,ruled,vlined]{algorithm2e}
\usepackage{enumerate}
\usepackage{graphicx}
\usepackage{enumitem}
\usepackage{amsmath}
\usepackage{algpseudocode}
\usepackage{amssymb}
\usepackage{times}
\usepackage{endnotes}
\usepackage{stfloats}
\usepackage{stkernel}
\usepackage{color}\usepackage{enumerate}
\usepackage{array}
\usepackage{relsize}
\usepackage[belowskip=-10pt,aboveskip=3pt,footnotesize]{caption}
\usepackage{cite}
\usepackage{textcomp}

\renewcommand{\baselinestretch}{0.98}

\newcommand{\Ncal}{\ensuremath{{\mathcal{N}}}}

\newcommand{\Ccal}{\ensuremath{{ \mathcal{C} }}}

\newcommand{\tr}[1]{\textrm{#1}}
\newcommand{\mr}[1]{\mathrm{#1}}

\newcommand{\jexp}[1]{e^{\jmath {#1} }}
\newcommand{\nnsum}[2]{\underset{#1}{\overset{#2}{\sum}}}
\newcommand{\nsum}[2]{{\sum_{#1}^{#2}}}

\newcommand{\nsigma}[2]{\sigma_{\mathrm{ #1}}^{{ #2} }}

\newcommand{\T}[1]{T_{\mr{ #1}}}

\newcommand{\limit}[2]{\underset{{#1} \to {#2}} \lim}

\newcommand{\temp}[1]{\mathrm{I}_{\mathrm{#1}}}

\newcommand{\wdn}[1]{w_{\mr{d},#1}}
\newcommand{\bwu}[1]{\mathbf{w}_{\mr{u},#1}}
\newcommand{\wun}[2]{{w}_{\mr{u},#1}^{(#2)}}

\newcommand{\minus}{\scalebox{0.5}[1.0]{$-$}}

\newcommand{\hhat}[1]{\hat{{h}}_{#1}}
\newcommand{\Hhat}[1]{\hat{\mathbf{H}}_{#1}}

\newcommand{\bN}[1]{\mathbf{N}_{#1}}
\newcommand{\bA}[1]{\mathbf{A}_{#1}}
\newcommand{\bU}[1]{\mathbf{U}_{#1}}
\newcommand{\bV}[1]{\mathbf{V}_{#1}}
\newcommand{\bG}[1]{\mathbf{G}_{#1}}

\newcommand{\bH}[1]{\mathbf{H}_{#1}}
\newcommand{\bh}[1]{\mathbf{h}_{#1}}
\newcommand{\bx}[1]{\mathbf{x}_{#1}}
\newcommand{\bw}[1]{\mathbf{w}_{#1}}
\newcommand{\h}[2]{h_{#1}^{(#2)}}
\newcommand{\bhhat}[1]{\hat{\mathbf{h}}_{#1}}

\newcommand{\pu}[1]{p_{\mr{u},#1}}

\newcommand{\bP}[1]{\mathbf{P}_{#1}}

\newcommand{\byu}[1]{\mathbf{y}_{\mr{u},#1}}

\newcommand{\yu}[2]{y_{\mr{u},#1}^{(#2)}}
\newcommand{\yd}[2]{y_{\mr{d},#1}^{#2}}

\newcommand{\bg}[1]{\mathbf{g}_{#1}}
\newcommand{\bcsym}[1]{\mathbf{c}_{#1}}
\newcommand{\csym}[1]{c_{#1}}

\newcommand{\phit}[2]{\phi_{#1}^{(#2)}}
\newcommand{\varphit}[2]{\varphi_{#1}^{(#2)}}

\newcommand{\bTheta}[1]{\mathbf{\Theta}_{#1}}
\newcommand{\thetat}[2]{\theta_{#1}^{(#2)}}

\newcommand{\squarel}{\left[}
\newcommand{\squarer}{\right]}
\newcommand{\curll}{\left\{}
\newcommand{\curlr}{\right\}}
\newcommand{\absl}{\left|}
\newcommand{\absr}{\right|}
\newcommand{\parl}{\left(}
\newcommand{\parr}{\right)}

\newtheorem{prop}{Proposition}
\newtheorem{thm}{Theorem}
\newtheorem{lem}{Lemma}
\newtheorem{rem}{Remark}

\newtheorem{corr}{Corollary}
\newtheorem{definition}{Definition}

\begin{document}

\title{Linear Massive MIMO Precoders in the Presence of Phase Noise -- A Large-Scale Analysis}

\author{R.~Krishnan,~M.~R.~Khanzadi,~N.~Krishnan,~\IEEEmembership{Member,~IEEE,}~Y.~Wu,~\IEEEmembership{Member,~IEEE,}~A.~Graell~i~Amat,~\IEEEmembership{Senior~Member,~IEEE,}~T.~Eriksson,~and~R.~Schober,~\IEEEmembership{Fellow,~IEEE}

\thanks{Rajet Krishnan, M. R. Khanzadi, Alexandre Graell i Amat, and Thomas Eriksson are with the Department
of Signals and Systems, Chalmers University of Technology, Gothenburg, Sweden (e-mail: \{rajet, khanzadi, alexandre.graell, thomase\}@chalmers.se).  Y. Wu and R. Schober are with the Institute for Digital Communications Friedrich-Alexander University Erlangen-Nurnberg (FAU) (e-mail: \{yongpeng.wu, schober\}@lnt.de).  N. Krishnan is with Qualcomm, San Diego, California, USA (e-mail: nakrishn@qti.qualcomm.com).}%
\thanks{Research supported by the Swedish Research Council under grant \#2011-5961.}
}

%


\markboth{IEEE Transactions on Vehicular Technology}%
{Submitted paper}

\maketitle

\begin{abstract}
We study the impact of phase noise on the downlink performance of a multi-user multiple-input multiple-output system, where the base station (BS) employs a large number of transmit antennas $M$. We consider a setup where the BS employs $M_{\mr{osc}}$  free-running oscillators, and $M/M_{\mr{osc}}$  antennas are connected to each oscillator. For this configuration, we analyze the impact of phase noise on the  performance of the zero-forcing (ZF), regularized ZF, and matched filter (MF) precoders when $M$ and the number of users $K$ are asymptotically large, while the ratio $M/K=\beta$ is fixed. We analytically show that the impact of phase noise on the signal-to-interference-plus-noise ratio (SINR) can be quantified as an effective reduction in the quality of the channel state information available at the BS when compared to a system without phase noise.  As a consequence, we observe that as  $M_{\mr{osc}}$ increases, the SINR performance of all considered precoders degrades. On the other hand, the variance of the random phase variations caused by the BS oscillators reduces with increasing  $M_{\mr{osc}}$. Through Monte-Carlo simulations, we verify our analytical results, and compare the performance of the precoders for different phase noise and channel noise variances. For all considered precoders, we show that when $\beta$ is small, the performance of the setup where all BS antennas are connected to a single oscillator is superior to that of the setup where each BS antenna has its own oscillator. However, the opposite is true when $\beta$ is large and the signal-to-noise ratio at the users is low.

\emph{Index Terms} --  Massive MIMO, linear precoding, phase noise, broadcast channel, random matrix theory, multi-user MIMO.

\end{abstract}
\section{Introduction}
\label{sec:Intro}

\PARstart{M}{assive} multiple-input multiple-output (MIMO) is a promising technology for future wireless networks \cite{Rusek13,Marzetta10}. This technology deploys antenna arrays containing hundreds of antennas, which can be exploited to significantly enhance the network throughput and energy efficiency performances \cite{Larsson13}.

In particular, employing massive antenna arrays at the base station (BS) is expected to provide significant array gains and improved spatial precoding resolution for downlink transmission in multi-user (MU) MIMO systems \cite{Caire10}. This in turn is expected to increase the throughput per user equipment (UE), and enable the support of a large number of UEs at the same time. 

It is known that as the number of BS antennas $M$ becomes asymptotically large, the channels between the BS and the different UEs become approximately orthogonal \cite{Rusek13}, \cite{Edfors14}. This indicates that significant spatial diversity can be achieved due to the spatial separation between the antennas of the different UEs. Thus, the MU-massive-MIMO downlink channel, which is a MIMO broadcast channel (BC), is inherently robust to channel correlation effects \cite{Rusek13}.

In general, MIMO systems suffer from MU interference during downlink transmission, which is mitigated by means of channel-aware precoding methods implemented at the BS \cite{Caire03}. Nonlinear precoding methods such as dirty-paper coding are capacity achieving for the MIMO BC \cite{Caire03}. However, these precoders are highly complex, thereby motivating the need for computationally simpler methods such as linear precoders \cite{Goldsmith06}. For MIMO systems, in the asymptotic regime, where $M$ and the number of UEs, $K$, are asymptotically large, linear precoders have been shown to achieve close-to-optimal performance \cite{Hochwald02, Peel05, Wiesel06, Hanly12}.

One of the early works analyzing the downlink performance of MIMO in the asymptotic regime is \cite{Hochwald02}, where it is shown that for $M,K \rightarrow \infty$, a linear growth in the sum rate with $M$ and $K$ can be achieved with the zero-forcing (ZF) precoder. However, in \cite{Peel05}, it is shown that the linear sum rate growth cannot be achieved for the ZF precoder when $M/K = 1$. This issue is overcome by regularized ZF (RZF) precoding, whose regularization parameter $\alpha$ can be chosen such that it optimizes the signal-to-interference-plus-noise ratio (SINR) or the sum mean square error (sum MSE). In \cite{Wiesel06, Hanly12}, it is shown that, in the asymptotic limit, the RZF precoder maximizes the SINR in the case of independent identically distributed (i.i.d.) Gaussian channels. In \cite{Hoydis13}, it is shown that for $M,K \rightarrow \infty$, the matched filter (MF) precoder requires significantly larger numbers of antennas than the RZF precoder to achieve the same performance. However, \cite{Marzetta13_1} reveals that the MF precoder outperforms the ZF precoder in terms of energy efficiency, while in terms of spectral efficiency, the ZF precoder performs better. In some  recent works \cite{Debbah12, Emil13_1}, the SINRs achieved by linear precoders in the asymptotic regime are derived for the case where the channels between the BS and UEs are correlated. In addition, imperfect channel state information (CSI) at the BS is assumed, and the precoders are optimized such that the SINR achieved at the UEs is maximized.

In most prior works on MU-MIMO downlink transmission \cite{Goldsmith06,Hochwald02, Peel05, Wiesel06, Hanly12, Hoydis13,Marzetta13_1, Debbah12, Emil13_1}, it is assumed that the hardware components of the MIMO transceiver are ideal. However, it is now well understood that the performance of these systems can be severely limited by impairments arising from nonideal transceiver hardware components \cite{Emil13}. Furthermore, implementing linear precoding methods at the BS mandates the availability of reliable CSI. 
This is challenging since the coherence time of the channels between the BS and its associated UEs is finite, and thus the BS is required to update its CSI regularly. Also, hardware impairments affect the CSI quality drastically, and the phase noise caused by noisy local oscillators used in the transceivers is a major contributor to this problem  \cite{Emil13, Rajet14, Larsson12}. In general, phase noise manifests itself as a random, time-varying phase difference between the oscillators connected to the antennas at the BS and the UEs \cite{Colavolpe05,Rajet13}. Phase noise causes random rotations of the transmitted data symbols, thereby causing performance degradation. As illustrated in \cite{Rajet14,Larsson12}, phase noise also causes partial coherency loss, i.e., the true channel during the data transmission period can become significantly different from the CSI acquired during the training period. This is referred to as the \textit{channel-aging phenomenon} \cite{Heath13}. 

It is therefore expected that phase noise  at the BS and the UEs will present a serious challenge towards realizing the unprecedented advantages promised by massive MIMO \cite{Rusek13}. The effect of phase noise on the uplink performance of a massive MIMO system has been analyzed in \cite{Rajet14,Larsson12,Emil13}. One of the early works on the downlink performance of MIMO systems in the presence of phase noise is \cite{Beamforming10}, where the error-vector magnitude degradation is analyzed. However, the number of studies on the impact of phase noise on the downlink performance of a massive MIMO system are limited. Prior work on the downlink performance of massive MIMO systems  confirms that their performance can be reliably predicted using large-scale analysis, where $M,K\rightarrow \infty$  \cite{Debbah12,Emil13_1}. This is because, even though the SINR depends on the instantaneous values of the channel and other random effects, these effects become deterministic in the asymptotic regime \cite{Debbah12,Emil13_1}. Interestingly, large-scale analysis of MIMO systems is accurate even for practical values of $M$ and $K$. Hence, we expect that a similar analysis using tools from random matrix theory (RMT) \cite{Debbah12,Tao,Speicher} will provide new and important insights on the impact of phase noise on the performance of  massive MIMO downlink transmissions, and this constitutes one of the main motivations for  this work.

In this paper, we analyze the massive MIMO downlink performance of linear precoding schemes, including the ZF, RZF, and MF precoders in the presence of oscillator phase noise. We consider a single-cell massive MIMO system comprising one BS serving multiple single-antenna UEs. We analyze a general setup, shown in Fig. \ref{fig:GOSetup}, where the BS employs $M_{\mr{osc}}$  free-running oscillators, and $M/M_{\mr{osc}} \in \mathbb{Z}^+$ BS antennas are connected to each oscillator. We refer to this as the general oscillator (GO) setup. Two interesting special cases arise from this general setup. In the first case, all BS antennas are connected to a single oscillator (referred to as the common oscillator (CO) setup). In the second case, each BS antenna has its own oscillator (referred to as the distributed oscillator (DO) setup). For the considered setups, we obtain the following results:
\begin{itemize}[leftmargin=*]
\item For the GO setup, we derive the effective SINR \cite{Marzetta13_1, Michailis14} at a given UE  for the RZF precoder as  $M,K\rightarrow \infty$, while the ratio $M/K=\beta$ is fixed. Then, we derive the optimal regularization parameter $\alpha$, which maximizes the effective SINR for the RZF precoder. Furthermore,  we derive the effective SINRs of the ZF and MF precoders for the GO setup, by treating them as special cases of the RZF precoder.
\item We show that the impact of phase noise on the SINR of the precoders can be quantified as an effective reduction of the quality of the CSI available at the BS, when compared to the system without phase noise. As a consequence,  we observe that as  $M_{\mr{osc}}$ increases, the SINR performance of all considered precoders degrades. Specifically, the desired signal power decreases as $M_{\mr{osc}}$ increases  for all considered precoders. Furthermore, the interference power increases with increasing $M_{\mr{osc}}$ for the RZF and the ZF precoders. For the MF precoder, the interference power is almost independent from the effect of oscillator phase noise. However, the variance of the random phase variations caused by the BS oscillators reduces as $M_{\mr{osc}}$ increases.
\item We compare the performance of the precoders for different phase noise and additive white Gaussian noise (AWGN) variances by using Monte-Carlo (MC) simulations. We show that the SINRs derived for the precoders are accurate for relevant and  practical values of $M$ and $K$. Furthermore, we illustrate how the  relative performance of the precoders depends on $M_{\mr{osc}}$, $M$, $\beta$, the phase noise and AWGN variances, and the CSI quality at the BS.
\item Finally, we compare the achievable rates of the CO and the DO setups via simulations. A general observation for all considered precoders is that the CO setup performs better than the DO setup when $\beta$ is small.  However, the opposite is true when $\beta$ is large and the signal-to-noise ratio (SNR) at the UE is low.
\end{itemize}

The remainder of the paper is organized as follows. In Section \ref{sec:SysMod}, we introduce the massive MIMO system model with phase noise and AWGN, and review the time-division duplexed (TDD) transmission mode and the considered linear precoders. We present a large-scale analysis of the received signal, the effective SINR, and other analytical results in Sections \ref{sec:RxdAna} and \ref{sec:MainResults}. In Section \ref{sec:Sim}, we discuss our analytical and simulation results. We summarize our key findings in Section \ref{sec:Conc}. Some useful results from the literature are presented in Appendix A, and the proofs for the main analytical results are provided in Appendices B and C.

\emph{Notation}: Vectors and matrices are represented by boldface lower-case and bold-face upper-case letters, respectively. The complex Gaussian distribution and the real Gaussian distribution with mean $\mu$ and variance $\sigma^2$ are denoted as $\mathcal{CN}(\mu,\sigma^2)$, and $\mathcal{N}(\mu,\sigma^2)$, respectively. The Hermitian, conjugation, expectation with respect to $\phi$, and trace operators are denoted as $ \curll \cdot \curlr^{\mr{H}}, \curll \cdot \curlr^{*}, \mathbb{E}_{\phi} \curll \cdot \curlr, ~\tr{and}~ \mr{tr}\curll \cdot \curlr$,  respectively. $\Re\{\cdot\}$, $\Im\{\cdot\}$, $| \cdot |$ and $\angle \cdot$ are the real part, imaginary part, magnitude, and the angle of a complex number, respectively. The $\mr{L}2$ norm of a vector is denoted by $\lVert {\cdot} \rVert$. An $n$-dimensional complex vector is denoted by $\mathbb{C}^{n\times 1}$, while $\mathbb{C}^{n\times m}$ denotes the generalization to an $(n \times m)$-dimensional complex matrix. The integer space is denoted by $\mathbb{Z}$. The $\mr{diag}\curll \ldots \curlr$ operator generates a diagonal matrix from a given vector, while $\mathbf{I}_M$ denotes an $M \times M$ identity matrix.

\section{System Model}
\label{sec:SysMod}

In this section, we  introduce the considered single-cell massive MIMO system with  i.i.d.  flat-fading channels, oscillator phase noise at the BS and the UEs, TDD operation, channel estimation at the BS, and linear precoding.

\subsection{Channel and Phase Noise Models}

We consider a single-cell system consisting of a BS that serves $K$ UEs. The BS is equipped with $M$ antennas and each UE is equipped with a single antenna. The channel between the $m$th antenna, $m \in \{1,\ldots, M\}$,  at the BS and the $k$th UE, $k \in \{1,\ldots, K\}$, is assumed to be frequency-flat Rayleigh block-fading,  and its gain is denoted as $\h{k}{m} \sim \mathcal{CN}(0,\nsigma{\mr{ch}}{2})$, where $\nsigma{\mr{ch}}{2}=1$. Furthermore, $\h{k}{m}$ is the $(k,m)$th entry of $\bH{}\in \mathbb{C}^{K \times M}$, which represents the channels between the BS and the UEs. The coherence time of the block-fading channel is denoted by $\T{c}$. Without loss of generality, it is assumed that the large-scale fading component of the channel is unity.

\begin{figure}[!t]
\begin{center}
\includegraphics[width = 3.0in, keepaspectratio=true]{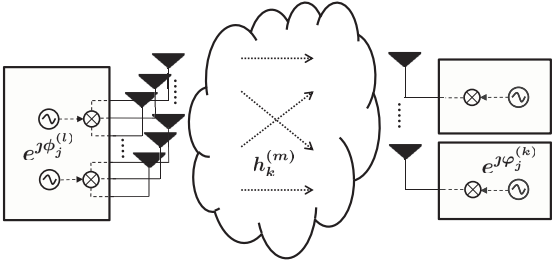}
\caption{The general oscillator (GO) setup, where the BS has $M_{\mr{osc}}$  free-running oscillators, and $M/M_{\mr{osc}} \in \mathbb{Z}^+$ BS antennas are connected to each oscillator.}
\label{fig:GOSetup}
\end{center}
\end{figure}

We consider the GO setup where the BS employs $M_{\mr{osc}}$  free-running oscillators, and $M/M_{\mr{osc}} \in \mathbb{Z}^{+}$ BS antennas are connected to each oscillator. Considering a discrete-time Wiener phase noise model \cite{Colavolpe05, Rajet13}, in the $j$th symbol interval, the phase noise sample at the $k$th UE is denoted by $\varphit{j}{k}$, and that of the $l$th oscillator at the BS is denoted by $\phit{j}{l}$ \cite{Demir00}, where
\begin{IEEEeqnarray}{rCl}
\label{eq:SysMod1}
\varphit{j}{k}  &=& \varphit{j-1}{k} + \Delta_{j}^{\varphi},\quad \Delta_{j}^{\varphi} \sim \mathcal{N}(0,\nsigma{\varphi}{2}),\\ \label{eq:wiener}
\phit{j}{l}  &=& \phit{j-1}{l} + \Delta_{j}^{\phi},\quad \Delta_{j}^{^{\phi}} \sim \mathcal{N}(0,\nsigma{\phi}{2}).
\end{IEEEeqnarray}
Here, $k \in \{1,\ldots, K\}$, $l \in \{1,\ldots, M_{\mr{osc}}\}$, and $\nsigma{\varphi}{2}$ and $\nsigma{\phi}{2}$ denote the phase noise increment variances at the UE and the BS, respectively. Since the channel is constant within $\T{c}$, and given that the $m$th BS antenna is connected to the $l$th oscillator at the BS, $\thetat{j,k}{m} \triangleq \varphit{j}{k} + \phit{j}{l}$ is the phase noise sample that impairs the link between the $k$th UE and the $m$th BS antenna. 


\subsection{TDD and Channel Estimation}

The TDD mode of operation is assumed, where the UEs first transmit orthogonal pilots to the BS in order to facilitate channel estimation. Upon reception of the pilots from the UEs, the BS forms an estimate of the channel between its antennas and the UEs. Exploiting channel reciprocity \cite{Marzetta10}, the channel estimate is then used by the BS to transmit data in the downlink to the UEs.

During channel training, the UEs transmit uplink pilot symbols that are orthogonal in time. Specifically, the UEs transmit their pilot symbols sequentially in time, meaning that when one UE is transmitting, the other $K-1$ UEs are silent. This training scheme permits a simple channel estimation method at the BS, and as we shall see later, facilitates the large-scale analysis of the MIMO downlink performance \cite{Larsson12,Rajet14,Marzetta10,Debbah12,Emil13_1}. Importantly, this training scheme aids in capturing the channel-aging effect due to phase noise on the SINR achieved at the UEs. Our analysis of the impact of phase noise on the system performance for this specific training scheme qualitatively also extends to other training schemes---see \cite{Rajet14}, where a code-orthogonal training scheme is also considered.

The signal received at the BS from the $k$th UE at time instant $j=0$ can be written as
\begin{IEEEeqnarray}{rCl}
\label{eq:SysMod2}
\byu{0} =   \sqrt{\pu{k}} \bTheta{0,k} \bh{k}\csym{0,k} + \bwu{0},
\end{IEEEeqnarray}
where perfect timing and frequency synchronization are assumed~\cite{Colavolpe05,Rajet13}. In \eqref{eq:SysMod2}, $\byu{0} = [\yu{0}{1}, \ldots, \yu{0}{M}]^{\mr{T}}$, where $\yu{0}{m}$ represents the received signal in the uplink at the $m$th BS antenna. ${\pu{k}}$ denotes the uplink transmit power of the $k$th UE,  $\bTheta{0,k} = \mr{diag} \curll \jexp{\thetat{0,k}{1}}  \mathbb{\mathbf{1}}^{\mr{T}}_{1 \times {M}/{M_{\mr{osc}}}}, \ldots  ,\jexp{\thetat{0,k}{M_{\mr{osc}}}}   \mathbb{\mathbf{1}}^{\mr{T}}_{1 \times {M}/{M_{\mr{osc}}}}, \curlr$, where $\mathbb{\mathbf{1}}^{\mr{T}}_{1 \times {M}/{M_{\mr{osc}}}}$ denotes an all-one vector of length ${M}/{M_{\mr{osc}}}$, and $\bh{k} = [\h{k}{1}, \ldots, \h{k}{M}]^{\mr{T}}$. $\csym{0,k}$ denotes the pilot symbol transmitted by the $k$th UE. $\bwu{0} = [\wun{0}{1}, \ldots, \wun{0}{M}]^{\mr{T}}$, where $\wun{j}{m} \sim \Ncal(0,\nsigma{w}{2})$ denotes the zero-mean AWGN random variable (RV) at the $m$th receive antenna.
%

Based on $\byu{0}$, a linear MMSE channel estimate is formed for the $k$th UE at time instant $j=0$. This estimate can modeled as a Gauss-Markov process \cite{Debbah12,Heath11}
\begin{IEEEeqnarray}{rCl}
\label{eq:SysMod3}
\bhhat{0,k} = \sqrt{q_{0,k}}\bTheta{0,k} \bh{k}  + \sqrt{q_{1,k}} \bw{\mr{e},k}.
\end{IEEEeqnarray}
Here, $ \bw{\mr{e},k} \in \mathbb{C}^{M\times 1}$ represents the estimation error, and its entries are complex Gaussian i.i.d. RVs with zero mean and unit variance. $\bhhat{0,k}\in \mathbb{C}^{M\times 1}$ in \eqref{eq:SysMod3} is the $k$th row of $\Hhat{0}\in \mathbb{C}^{K \times M}$, which contains the channel estimates of all UE channels. For the linear MMSE, The entries of $ \bw{\mr{e},k}$ and $\bhhat{0,k}$ are considered to be statistically independent of each other \cite{Heath11}. Without loss of generality, we assume that $q_{0} = {q_{0,k}}$ and ${q_{1}} = {q_{1,k}}$, and set the channel estimate variance as $\nsigma{\hhat{k}}{2} \triangleq {q_{0}}\nsigma{\mr{ch}}{2} + {q_{1}} = 1$. The parameter ${q_{0}} \in  [0, 1]$ reflects the quality of the  channel estimate. When ${q_{0}}  = 1$, a perfect channel estimate is available at the BS, while for ${q_{0}} = 0$, the channel estimate is completely uncorrelated with respect to the original channel.

\subsection{Downlink Transmission and Linear Precoding}

Let the data transmission on the downlink from the BS to the UEs commence in the symbol interval  $j=\tau$, where ($\tau < \T{c}$). Here, $\tau$ denotes the symbol periods that have elapsed after the uplink transmission of the pilot symbols from the $k$th UE. The signal received by the $k$th UE  at time instant $\tau$ can be expressed    as
\begin{IEEEeqnarray}{rCl}
\label{eq:SysMod41}
\yd{\tau}{k} &=&  \bh{\tau,k}^{\mr{T}} \bTheta{\tau,k} \bx{\tau}  +  \wdn{k} \\
\label{eq:SysMod42}
&=&  \bh{\tau,k}^{\mr{T}} \bTheta{\tau,k} \nsum{k_1=1}{K} \sqrt{p_{k_1}} \bg{0,k_1} \csym{\tau,k_1}  +  \wdn{k}.
\end{IEEEeqnarray}
In \eqref{eq:SysMod41}, $ \bx{\tau}\in \mathbb{C}^{M\times 1}$ is the transmit signal,  and in \eqref{eq:SysMod42}, $\bx{\tau}$ is written as a linear combination of the data symbols $\csym{\tau,k_1}, k_1 \in \{1,\ldots,K\}$, transmitted to the $K$ UEs. The data symbols are assumed to be circularly symmetric, but not necessarily Gaussian distributed \cite{Belzer01}. $\bg{0,k}\in \mathbb{C}^{M\times 1}$ is the $k$th column of the downlink precoding matrix $\bG{0}\in \mathbb{C}^{M\times K}$, where  $\bG{0}= [ \bg{0,1}, \ldots,  \bg{0,K}]$. $ \wdn{k} \sim \mathcal{CN}(0,\nsigma{w}{2})$ denotes the AWGN RV at the $k$th UE.

In this work, we consider RZF precoding in \eqref{eq:SysMod42}, i.e.,   $\bG{0}$ can be written as \cite{Debbah12}
\begin{IEEEeqnarray}{rCl} \label{eq:SysMod51}
\bG{0} = \xi \parl \Hhat{0}^{\mr{H}}\Hhat{0} + M \alpha \mathbf{I}_M \parr^{\minus 1} \Hhat{0}^{\mr{H}} {\bP{}}^{\frac{1}{2}},
\end{IEEEeqnarray}
where $\bP{} \triangleq \mr{diag}\curll p_1,\ldots,p_K \curlr$,  $p_{k}, \forall k \in \{1,\ldots,K\}$ denotes the power allocated to the $k$th UE, and the normalization parameter $\xi$ is set such that the precoder satisfies the power constraint $\mr{tr} \parl \bG{0}^{\mr{H}} \bG{0}   \parr  = 1$. Note that the RZF precoder in \eqref{eq:SysMod51} is known to perform better than the other linear precoders, such as the MF and the ZF precoders \cite{Peel05, Hanly12}. Furthermore, the RZF precoder simplifies to the ZF precoder when $\alpha \rightarrow 0$, i.e., $\bG{0} = \xi \Hhat{0}^{\mr{H}} \parl \Hhat{0}\Hhat{0}^{\mr{H}} \parr^{\minus 1}  {\bP{}}^{\frac{1}{2}}$, and to the MF precoder when $\alpha \rightarrow  \infty$, i.e., $\bG{0} = \xi \Hhat{0}^{\mr{H}}   {\bP{}}^{\frac{1}{2}}$  \cite{Debbah12}.

\section{Large-Scale Analysis of the Received Signal and Achievable Rates}
\label{sec:RxdAna}

In this section, we use tools from RMT to analyze the received signal model in  \eqref{eq:SysMod42}. Specifically, we present a simplification of the desired signal term  in  \eqref{eq:SysMod42} for the GO setup when  $M,K\rightarrow \infty$, while $M/K=\beta$. Notably, we will show that in the CO and DO setups, the  multiple-input single-output (MISO) system model in \eqref{eq:SysMod42} can be re-written as an equivalent single-input single-output (SISO) phase noise channel including the effects of phase noise, AWGN, and interference \cite{Khanzadi14}. Furthermore, we define the effective SINR, and discuss the achievable rates  for the GO setup.

\subsection{Received Signal Model}

For the RZF precoder in \eqref{eq:SysMod51}, the received signal at the $k$th UE in \eqref{eq:SysMod42} becomes
\begin{IEEEeqnarray}{rCl}
\label{eq:SysMod61}
\yd{\tau}{k} &=&  \bh{k}^{\mr{T}} \bTheta{\tau,k} \bG{0} \bcsym{\tau}  +  \wdn{k} \nonumber\\
\label{eq:SysMod64}
&=& \underbrace{\sqrt{p_{k}} \bh{k}^{\mr{T}} \bTheta{\tau,k} \xi \parl \Hhat{0}^{\mr{H}}\Hhat{0} + M \alpha \mathbf{I}_M \parr^{\minus 1} \bhhat{0,k}^{*}}_{\triangleq \temp{sig}} \csym{\tau,k}   \nonumber\\ &&
+ \underbrace{\bh{k}^{\mr{T}} \bTheta{\tau,k}  \xi \parl \Hhat{0}^{\mr{H}}\Hhat{0} + M \alpha \mathbf{I}_M \parr^{\minus 1} \Hhat{0,\minus k}^{\mr{H}} \bP{\minus k}^{\frac{1}{2}} }_{\triangleq \boldsymbol{\temp{}}_{\mr{int}}}\bcsym{\tau,\minus k}\nonumber\\
&& +  \wdn{k}.
\end{IEEEeqnarray}
In \eqref{eq:SysMod64}, we have introduced the following definitions: $ \Hhat{0,\minus k} =$ $ [\bhhat{0,1},\ldots,\bhhat{0,k-1},\bhhat{0,k+1},\ldots, \bhhat{0,K}]$, $\bP{\minus k} = $ $\mr{diag}\curll p_1,\ldots,p_{k-1},p_{k+1},\ldots, p_K \curlr$, and $\bcsym{\tau,\minus k} =$ $ [\csym{\tau,1}, \ldots, \csym{\tau,k-1},\csym{\tau,k+1},\ldots,\csym{\tau,K} ]^{\mr{T}}$. Furthermore, $\temp{sig}\in \mathbb{C}$ and $\boldsymbol{\temp{}}_{\mr{int}}^{\mr{T}} \in \mathbb{C}^{M-1 \times 1}$ denote the scaling factors associated with the desired symbol and the interfering symbols at the $k$th UE, respectively. The factor $\temp{sig}$ is simplified in the following proposition.

\begin{prop}
\label{prop:Irzfsig}
Consider an RZF precoded downlink transmission from a BS having $M$ antennas to $K$ single-antenna UEs employing TDD in the presence of oscillator phase noise. Let $\alpha > 0 $, $M/K = \beta, \beta \geq 1$, and $q_0  \in [0,1]$. Assume that $\frac{1}{M}\Hhat{}^{\mr{H}}\Hhat{}$ has uniformly bounded spectral norm for all $M$. Then, the desired signal factor $\temp{sig}$, for $M,K \rightarrow \infty$,  can be simplified  to
\begin{IEEEeqnarray}{rCl} \label{eq:SysMod71}
&&\temp{sig}=  {\sqrt{{p_{k}}q_0   } \T{PN}  \xi   t    }    e^{\jmath (\varphit{\tau}{k} - \varphit{0}{k})  }, \IEEEeqnarraynumspace
\end{IEEEeqnarray}
where
\begin{IEEEeqnarray}{rCl}
\label{eq:SysMod81a}
\T{PN} &\triangleq& \limit{M}{\infty} \frac{1}{M}   \mr{tr} \curll \Delta \boldsymbol{\Phi}_{\tau}\curlr, \\
\label{eq:SysMod81aa}
\xi &=& \limit{M,K}{\infty} \sqrt{ \frac{M{\parl 1 +     m(-\alpha) \parr^2 }}{ m'(-\alpha)\nsum{k=1}{K}   {p_k   } } } \\
\label{eq:SysMod81b}
m(-\alpha) &=& \frac{   \beta- 1 - \alpha\beta + \sqrt{\beta^2\alpha^2 + 2(\beta + 1)\alpha\beta + (1 - \beta)^2} }{2\alpha\beta} \IEEEeqnarraynumspace \\
\label{eq:SysMod81c}
t &=& \frac{m(-\alpha)}{m(-\alpha) + 1}.
\end{IEEEeqnarray}
In \eqref{eq:SysMod81a} $\Delta  \boldsymbol{\Phi}_{\tau} = \mr{diag} \curll \right.$ $\jexp{\phit{\tau}{1} - \phit{0}{1}}   \mathbb{\mathbf{1}}^{\mr{T}}_{1 \times {M}/{M_{\mr{osc}}}}$, $\ldots,$ $\left.\jexp{\phit{\tau}{M_{\mr{osc}}} - \phit{0}{M_{\mr{osc}}}}    \mathbb{\mathbf{1}}^{\mr{T}}_{1 \times {M}/{M_{\mr{osc}}}} \curlr$, $ m(-\alpha)$ in \eqref{eq:SysMod81b} is the Stieltjes Transform of the Marchenko-Pastur Law  \cite[Eqs. (1.12, 2.43)]{VerduRMT}, and $ m'(-\alpha) = \frac{dm(z)}{dz}|_{z = -\alpha}$.
\end{prop}
\begin{proof}
Please refer to Section B of Appendix C.
\end{proof}

\begin{rem}
The terms $t$ and $\xi$ in \eqref{eq:SysMod71} depend on $\alpha$ and $\beta$, and captures the channel hardening effect  \cite{Hochwald02,Marzetta10,Debbah12} that results from the averaging of the random fading channels when  RZF precoding is used, and $M, K \rightarrow \infty$. The term $\T{PN}$ in \eqref{eq:SysMod81a} captures the effects of phase noise variations  at the BS between the training and the data transmission phases, and is given by
\begin{IEEEeqnarray}{rCl}
\label{eq:SysMod81d}
\T{PN} = \frac{1}{M_{\mr{osc}}} \nnsum{l=1}{M_{\mr{osc}}} e^{\jmath (\phit{\tau}{l} - \phit{0}{l})}.
\end{IEEEeqnarray}
Specifically, for the CO setup, where $\Delta \boldsymbol{\Phi}_{\tau} =    e^{ \jmath(\phi_{\tau} - \phi_{0}{})  }\mathbf{I}_M$, and  the DO setup, where $ \Delta  \boldsymbol{\Phi}_{\tau} = \mr{diag} \curll \jexp{\phit{\tau}{1} - \phit{0}{1}}  , \ldots  , \jexp{\phit{\tau}{M} - \phit{0}{M}}   \curlr$, \eqref{eq:SysMod81d} reduces to  \cite{Rajet14}
\begin{IEEEeqnarray}{rCl}
\label{eq:SysMod91}
&&\T{PN}  \overset{M\rightarrow \infty}{\longrightarrow}  \left\{ \begin{array}{llll}
              e^{ \jmath(\phi_{\tau} - \phi_{0}{})  } & \mbox{CO setup}\\
            e^{-\frac{\tau \nsigma{\phi}{2}}{2}} & \mbox{DO setup}\end{array} \right.. %
\end{IEEEeqnarray}
Inspection of $\T{PN}$ in \eqref{eq:SysMod81d} and \eqref{eq:SysMod91}, reveals that $\angle \T{PN}$ reflects the random phase variations caused by the oscillators at the BS. The variance of $\angle \T{PN}$ decreases as $M_{\mr{osc}}$ increases, while its mean  is zero for all values of $M_{\mr{osc}}$. $|\T{PN}|$ represents the random amplitude variations in $\temp{sig}$ caused by  transmissions using distributed (asynchronous) oscillators at the BS. As $M_{\mr{osc}}$ increases, the mean of $|\T{PN}|$ reduces from $1$ to $\exp {\parl -\frac{\tau \nsigma{\phi}{2}}{2} \parr}$. Specifically, for the CO setup, $|\T{PN}|=1$, while for the DO setup, $|\T{PN}| = \exp {\parl -\frac{\tau \nsigma{\phi}{2}}{2} \parr}$.  The variance of $|\T{PN}|$ is the highest for $M_{\mr{osc}}=2$, and decreases as $M_{\mr{osc}}$ increases. In summary, when  $2 \leq M_{\mr{osc}}<\infty$, there are random variations in $\T{PN}$.  
In the DO setup\footnote{It is important to note that the results for the DO setup also hold in the case where $M_{\mr{osc}} \rightarrow \infty$, while the ratio  $M/M_{\mr{osc}}$ is fixed.}, where $M_{\mr{osc}} = M$, $\T{PN}$ hardens to a deterministic value that depends on $\tau$ and  $\nsigma{\phi}{2}$. However, this hardening effect is not observed in the CO setup as the phase noise caused by the BS oscillator does not average out.
\end{rem}

From the central limit theorem (as $K \rightarrow \infty$), and since the symbols $\bcsym{\tau}$ are circularly symmetric, the interference term $\boldsymbol{\temp{}}_{\mr{int}}\bcsym{\tau,\minus k}$ in \eqref{eq:SysMod64} is a circularly symmetric Gaussian RV for both the CO and the DO setups. Furthermore, this term is uncorrelated (and hence independent) from the signal term. In the case where $2\leq M_{\mr{osc}} < \infty$, $\boldsymbol{\temp{}}_{\mr{int}}\bcsym{\tau,\minus k}$ is non-Gaussian, but still  uncorrelated from the signal term. Further analysis of $\boldsymbol{\temp{}}_{\mr{int}}\bcsym{\tau,\minus k}$ is relegated to Section \ref{sec:MainResults}, and Appendices B and C. Upon applying \eqref{eq:SysMod71}, \eqref{eq:SysMod91} in \eqref{eq:SysMod64}, we have
\begin{IEEEeqnarray}{rCl}
\label{eq:SysMod91a}
\yd{\tau}{k} &=&  {\sqrt{{p_{k}}q_0   } \T{PN}  \xi   t    }    e^{\jmath (\varphit{\tau}{k} - \varphit{0}{k})  } \csym{\tau,k}  +  \boldsymbol{\temp{}}_{\mr{int}} \bcsym{\tau,\minus k} + \wdn{k}. \IEEEeqnarraynumspace
\end{IEEEeqnarray}
For the CO and the DO setups, the MISO system model in \eqref{eq:SysMod42} and \eqref{eq:SysMod64} becomes an equivalent SISO phase noise channel  in \eqref{eq:SysMod91a} \cite{Khanzadi14}. However, when $2 \leq M_{\mr{osc}} < \infty$, \eqref{eq:SysMod91a} still corresponds to a MISO phase noise channel, since $\T{PN}$ in \eqref{eq:SysMod81d} depends on the random phase noise variations of  the multiple oscillators at the BS.

\subsection{Effective SINR and Achievable Rates}

For the CO and the DO setups, we define the effective SINR based on the SISO phase noise channel in \eqref{eq:SysMod91a} as \cite{Khanzadi14}
\begin{IEEEeqnarray}{rCl}\label{eq:AchRate1}
\mathsf{SINR}_k =  \frac{  |\temp{sig}|^2}{  {\lVert \boldsymbol{\temp{}}_{\mr{int}} \rVert}^2 + \nsigma{w}{2}}. \IEEEeqnarraynumspace
\end{IEEEeqnarray}
Since the phase noise in \eqref{eq:wiener} drifts symbol-by-symbol, the SINR, which  depends on $\tau$, also varies symbol-by-symbol. In order to analyze the achievable rate of a given UE based on the SINR in \eqref{eq:AchRate1}, we model the phase noise to be a constant within a block of symbols \cite{Noels03}. Therefore, the SINR which is computed for a given $\tau$ corresponds to the SINR associated with a block of symbols, and can be used to determine the achievable rate \cite{Behrooz13}. {Note that this model is implicity used in \cite{Larsson13,Emil13, Rajet14}. Furthermore, this model is only used to evaluate the achievable rates based on the SINR derived, and is not required,  per se,  for deriving the SINR and the other analytical results in this paper.} The achievable rate computed based on this model is an upper-bound for the case where the SINR varies symbol-by-symbol.

Based on the effective SINR in \eqref{eq:AchRate1}, an upper bound  for the achievable rate of the $k$th UE  for the CO and the DO setups for a given block of symbols (i.e., given $\tau$) is  \cite{Khanzadi14}
\begin{IEEEeqnarray}{rCl}\label{eq:AchRate2}
C({\mathsf{SINR}_k})& \leq   \log_2{\parl 1 + \mathsf{SINR}_k\parr}. \IEEEeqnarraynumspace
\end{IEEEeqnarray}
This upper bound, which corresponds to the AWGN channel capacity, is generally tight for low-to-medium SINR values. Another upper bound for the achievable rate for the CO and the DO setups, which is generally tight at high SINR, was derived by Lapidoth \emph{et al.} \cite{Lapidoth02}, and is given by
\begin{IEEEeqnarray}{rCl}\label{eq:AchRate3}
C({\mathsf{SINR}_k}) \leq  \frac{1}{2}\log_2 (2 \pi \mathsf{SINR}_k)      - {\frac{1}{2} \log_2 \parl  { 2 \pi e  \tau(\nsigma{\varphi}{2} + \delta_{\mr{pn}}\nsigma{\phi}{2} )}  \ \parr}, \nonumber\\
\end{IEEEeqnarray}
where $\delta_{\mr{pn}}  =  1$, when $M_{\mr{osc}} = 1$, and $\delta_{\mr{pn}}  =  0$, otherwise. In \eqref{eq:AchRate3}, the second term represents the differential entropy of the phase noise process, $\varphit{\tau}{k} - \varphit{0}{k} + \delta_{\mr{pn}}(\phi_{\tau} - \phi_{0}{})$. The result in \eqref{eq:AchRate3} holds under the assumption that the phase-noise process is stationary, and has a finite differential-entropy rate. Combing \eqref{eq:AchRate2} and \eqref{eq:AchRate3}, the achievable rate for the CO and the DO setups can be tightly upper-bounded as $C(\mathsf{SINR}_k) \leq \min \{\mr{Rate~in~ \eqref{eq:AchRate2}}, \mr{Rate~in~ \eqref{eq:AchRate3}}\}$ \cite{Khanzadi14}.

There are no results available for the achievable rates for the MISO phase noise channel in  \eqref{eq:SysMod64}, when $2\leq M_{\mr{osc}} < \infty$. However, the randomness of $\T{PN}$ in this case is reminiscent of fading channels. Assuming ergodicity for the effective channel   in  \eqref{eq:SysMod91a}, the achievable rate  is written as \cite[Lemma 1]{Michailis14}
\begin{IEEEeqnarray}{rCl}\label{eq:AchRate4a}
C({\mathsf{SINR}_k}) &=& \mathbb{E}_\phi \log_2{\parl 1 + \frac{  |\temp{sig}|^2}{  {\lVert \boldsymbol{\temp{}}_{\mr{int}} \rVert}^2 + \nsigma{w}{2}} \parr} \\
\label{eq:AchRate4b}
&\approx&  \log_2{\parl 1 + \frac{ \mathbb{E}_\phi |\temp{sig}|^2}{ \mathbb{E}_\phi {\lVert \boldsymbol{\temp{}}_{\mr{int}} \rVert}^2 + \nsigma{w}{2}}\parr},
\end{IEEEeqnarray}
where $ \mathbb{E}_\phi$ denotes the expectation operation with respect to the phase noise at the BS. The accuracy of this approximation increases with increasing $M_{\mr{osc}}$. Also, $\temp{sig}$ and $\boldsymbol{\temp{}}_{\mr{int}}$ are not required to be independent. This motivates the definition of an effective SINR for $2\leq M_{\mr{osc}} < \infty$  as
\begin{IEEEeqnarray}{rCl}\label{eq:AchRate5}
\mathsf{SINR}_k =  \frac{ \mathbb{E}_\phi |\temp{sig}|^2}{ \mathbb{E}_\phi {\lVert \boldsymbol{\temp{}}_{\mr{int}} \rVert}^2 + \nsigma{w}{2}}. \IEEEeqnarraynumspace
\end{IEEEeqnarray}
The effective SINR in \eqref{eq:AchRate5} reduces to \eqref{eq:AchRate1} for the CO and the DO setups. Also, note that the achievable rate computed in \eqref{eq:AchRate4b} does not account for the differential entropy rates of the phase noise processes at the BS and the UEs.

%


\section{SINR Analysis and Optimal $\alpha$}
\label{sec:MainResults}
In this section, we first present the analytical results for the SINR achievable at a given UE for the considered precoders. We introduce Theorem \ref{thm:DetEqSINR}, which provides the effective SINR for the GO setup.  Then, the theorem is used to obtain the effective SINR when the ZF and the MF precoders are used at the BS. Furthermore, we analytically determine the optimal $\alpha$, which maximizes the effective SINR at a given UE for the RZF precoder.

\subsection{SINR of the RZF Precoder}
\begin{thm}
\label{thm:DetEqSINR}
Consider an RZF precoded downlink transmission from a BS having $M$ antennas to $K$ single-antenna UEs employing TDD in the presence of oscillator phase noise. Let $\alpha > 0$, $\beta \geq 1, q_0  \in  [0,1]$, and $\mathsf{SINR}_k$ denote the effective SINR at the $k$th UE. Then,
\begin{IEEEeqnarray}{rCl} \label{eq:EffSinr11}
 \mathsf{SINR}_k -  \mathsf{SINR}_{\mr{rzf}_k} \overset{M,K \rightarrow \infty}{\longrightarrow} 0
\end{IEEEeqnarray}
almost surely, and the effective SINR associated with the $k$th UE for the GO setup is given as
\begin{IEEEeqnarray}{rCl}
\label{eq:EffSinr21}
 \mathsf{SINR}_{\mr{rzf}_k} &=&  \frac{ {p_{k}  t^2  q_0 \mathbb{E}_\phi |\T{PN}|^2  } }{\frac{t_2}{M} \parl 1 -  {t q_0  \mathbb{E}_\phi |\T{PN}|^2 }   -  \frac{ t q_0 \mathbb{E}_\phi \absl \T{PN} \absr^2}{(1 + m(-\alpha))}   \parr + \frac{\nsigma{w}{2}}{\xi^2}}
\end{IEEEeqnarray}
with
\begin{IEEEeqnarray}{rCl}
\label{eq:EffSinr31a}
&&\mathbb{E}_\phi |\T{PN}|^2  \triangleq \mathbb{E}_\phi \absl \frac{1}{M}   \mr{tr} \curll \Delta \boldsymbol{\Phi}_{\tau}\curlr \absr ^2 = \frac{  1 - e^{-\tau \nsigma{\phi}{2} } }{M_{\mr{osc}}}  + e^{-\tau \nsigma{\phi}{2} } \\
\label{eq:EffSinr31b}
 &&t_2 = \nnsum{\stackrel{k_1=1,}{k_1 \neq k}}{K}   p_{k_1}   \frac{ m'(-\alpha)  }{( 1 +   m(-\alpha)    )^2},
\end{IEEEeqnarray}
where $t,m(-\alpha), \xi,$ and $\T{PN}$ are as given in \eqref{eq:SysMod81a}-\eqref{eq:SysMod81c}. Specifically, $\mathbb{E}_\phi |\T{PN}|^2 = \exp {\parl -{\tau \nsigma{\phi}{2}} \parr} $ for the DO setup, and  $\mathbb{E}_\phi |\T{PN}|^2 = 1$ for the CO setup.
\begin{proof}
Refer to Appendices B and C.
\end{proof}
\end{thm}

\begin{rem}
Theorem \ref{thm:DetEqSINR} captures the effect of phase noise on the SINR as an additional penalty to the quality of the channel estimate---when phase noise is present, the quality of the channel estimate, $q_0$, degrades to $q_0  \mathbb{E}_\phi |\T{PN}|^2$ in the GO setup \eqref{eq:EffSinr31a}. The quality of the effective channel estimate decreases as $M_{\mr{osc}}$ increases \eqref{eq:EffSinr31a}. Also, the effective quality is reduced when $\tau$ or $\nsigma{\phi}{2}$ increase. In the DO setup, the quality of the channel estimate diminishes by a factor of $\exp {\parl -{\tau \nsigma{\phi}{2}} \parr} $. However, in the CO setup, there is no reduction in the quality of the channel estimate due to noisy oscillators at the BS or the UEs as $|\T{PN}|^2 = 1$.

Clearly, the effect of channel aging on the SINR increases as $M_{\mr{osc}}$ increases---the SINR in \eqref{eq:EffSinr21} decreases as  $M_{\mr{osc}}$ increases. This results from the degradation of the desired signal power by a factor $\mathbb{E}_\phi |\T{PN}|^2 < 1$ when $M_{\mr{osc}}\geq2$, implying that the desired signal power decreases as $M_{\mr{osc}}$ increases. But,  the MU interference power increases with $M_{\mr{osc}}$, as can be seen in \eqref{eq:EffSinr21}. This is because the CSI quality deteriorates due to phase noise \eqref{eq:EffSinr31a}, thereby reducing the interference suppression capability of the RZF precoder. Specifically, to sum for the CO setup, there is no effect of phase noise on the SINR since the desired signal power and the MU interference power is the same as when only AWGN is present.
\end{rem}

In the sequel, we will use the SINR result in Theorem \ref{thm:DetEqSINR} in order to derive the effective SINR  for the ZF and MF precoders  (Corollaries \ref{corr:DetEqSinrZf} and \ref{corr:DetEqSinrMf}, respectively).


\subsection{SINR of the ZF Precoder}
\begin{corr}
\label{corr:DetEqSinrZf} 
Consider a ZF precoded transmission  from a BS having $M$ antennas to $K$ single-antenna UEs employing TDD in the presence of oscillator phase noise. Let $\alpha \rightarrow 0$, $\beta > 1, q_0  \in  [0,1]$, and assume that the minimum eigenvalue of $\frac{1}{M}\Hhat{}\Hhat{}^{\mr{H}}$ is bounded away from zero.  Then,
\begin{IEEEeqnarray}{rCl} \label{eq:DetEqSinrZf11}
 \mathsf{SINR}_k -  \mathsf{SINR}_{\mr{zf}_k} \overset{M,K \rightarrow \infty}{\longrightarrow} 0
\end{IEEEeqnarray}
almost surely, and the effective SINR associated with the $k$th UE in the GO setup is given by
\begin{IEEEeqnarray}{rCl}
\label{eq:DetEqSinrZf21}
 \mathsf{SINR}_{\mr{zf}_k} &=&  \frac{  p_k q_0   \mathbb{E}_\phi |\T{PN}|^2  }{\frac{t_2}{M} \parl 1 -  q_0  \mathbb{E}_\phi |\T{PN}|^2  \parr + \frac{\nsigma{w}{2}}{\xi^2}}, \IEEEeqnarraynumspace
\end{IEEEeqnarray}
where $ \T{PN}, \xi,$ and $t_2$ are as given in \eqref{eq:SysMod81a}, \eqref{eq:SysMod81b}, and \eqref{eq:EffSinr31b}.  Specifically, $\mathbb{E}_\phi |\T{PN}|^2 = \exp {-{\tau \nsigma{\phi}{2}}} $ for the DO setup, and  $\mathbb{E}_\phi |\T{PN}|^2 = 1$ for the CO setup.
\begin{proof}
When $\alpha$ is arbitrarily small, $ m(-\alpha) $ in \eqref{eq:SysMod81b} is arbitrarily large, i.e., $ m(-\alpha)  \gg 1$ \cite{VerduRMT}, and $ m(-\alpha)  + 1 \approx  m(-\alpha)$. Applying these approximations in the SINR of the RZF precoder in \eqref{eq:EffSinr21}, we directly obtain \eqref{eq:DetEqSinrZf21}.
\end{proof}
\end{corr}

As in the case of the RZF precoder, the SINR in \eqref{eq:DetEqSinrZf21} shows that the effect of channel aging is more pronounced for the GO setup when $M_{\mr{osc}}\geq2$. In the CO setup, the performance of the ZF precoder is not affected at all by the phase noise.

\subsection{SINR of the MF Precoder}

\begin{corr}
\label{corr:DetEqSinrMf}
Consider a MF precoded transmission  from a BS having $M$ antennas, to $K$ single-antenna UEs, employing TDD in the presence of oscillator phase noise. Let $M/K = \beta \geq 1, q_0  \in  [0,1]$, and  $\alpha \rightarrow \infty$. Then,
\begin{IEEEeqnarray}{rCl} \label{eq:DetEqSinrMf11}
 \mathsf{SINR}_k -  \mathsf{SINR}_{\mr{mf}_k} \overset{M,K \rightarrow \infty}{\longrightarrow} 0
\end{IEEEeqnarray}
almost surely, and the effective SINR associated with the $k$th UE in the GO setup is given by
\begin{IEEEeqnarray}{rCl}
\label{eq:DetEqSinrMf21}
 \mathsf{SINR}_{\mr{mf}_k} &=&  \frac{ M q_0 p_{k}   \mathbb{E}_\phi |\T{PN}|^2    }{  (\nsigma{w}{2}+1)\nsum{k=1}{K}    p_k }, \IEEEeqnarraynumspace
\end{IEEEeqnarray}
where $ \T{PN}$ is given in \eqref{eq:SysMod81a}. Specifically, $\mathbb{E}_\phi |\T{PN}|^2 = \exp {\parl -{\tau \nsigma{\phi}{2}}\parr } $ for the DO setup, and  $\mathbb{E}_\phi |\T{PN}|^2 = 1$ for the CO setup.
\begin{proof}
When $\alpha$ is arbitrarily large, $m(-\alpha)$ in \eqref{eq:SysMod81b} is arbitrarily small, i.e., $m(-\alpha) \ll 1$ and $m(-\alpha) + 1 \approx 1$. Applying these approximations in   \eqref{eq:EffSinr11}, we can write
\begin{IEEEeqnarray}{rCl}
\label{eq:DetEqSinrMf31}
 \mathsf{SINR}_{\mr{mf}_k} 
 \label{eq:DetEqSinrMf32}
&=&  \frac{  q_0 p_{k}  m(-\alpha)^2 \mathbb{E}_\phi \absl \T{PN} \absr^2    }{   m'(-\alpha) \nnsum{\stackrel{k_1=1,}{k_1 \neq k}}{K}   p_{k_1}     +  \nsigma{w}{2} m'(-\alpha)   \nsum{k=1}{K}    p_k     } \IEEEeqnarraynumspace\\
\label{eq:DetEqSinrMf33}
&=&  \frac{   q_0 p_{k}  m(-\alpha)^2 M \mathbb{E}_\phi \absl \T{PN} \absr    } {   m'(-\alpha)  \nsum{k=1}{K}    p_k \parl 1  +   \nsigma{w}{2}  \parr },   \IEEEeqnarraynumspace\\
\label{eq:DetEqSinrMf34}
&\approx&  \frac{   q_0 p_{k}  m(-\alpha)^2 M \mathbb{E}_\phi \absl \T{PN} \absr    } {  \parl \frac{2 m(-\alpha)}{\alpha} - \frac{1}{\alpha^2} \parr  \nsum{k=1}{K}    p_k \parl 1  +   \nsigma{w}{2}  \parr },  \IEEEeqnarraynumspace\\
\label{eq:DetEqSinrMf35}
&=& \frac{ M q_0 p_{k} \mathbb{E}_\phi \absl \T{PN} \absr    }{  (\nsigma{w}{2}+1)\nsum{k=1}{K}    p_k },  \IEEEeqnarraynumspace
\end{IEEEeqnarray}
where \eqref{eq:DetEqSinrMf32} is reduced to \eqref{eq:DetEqSinrMf33} by approximating $ \nsum{k=1}{K}    p_k  \approx  \nnsum{{k_1 \in K_{1}}}{}   p_{k_1}, K_1 = \{1, \ldots,k-1,k+1,\ldots, K\}$,  when $K \rightarrow \infty$. In \eqref{eq:DetEqSinrMf33}, we exploit that when $\alpha$ is arbitrarily large, $ m'(-\alpha)  \approx  \frac{2 m(-\alpha)}{\alpha} - \frac{1}{\alpha^2}$. Furthermore,  in  \eqref{eq:DetEqSinrMf34}, we exploit that as $\alpha \rightarrow \infty$, $\alpha m(-\alpha) \rightarrow 1$, which yields the final result in \eqref{eq:DetEqSinrMf35}.
\end{proof}
\end{corr}

\begin{rem}
From the SINR derived for the MF precoder in  \eqref{eq:DetEqSinrMf21}, it can be seen that the MU interference term is almost independent from the effect of phase noise at the BS in the GO setup (i.e., $M_{\mr{osc}}\geq 1$). This is unlike the effect of phase noise on the interference power of the RZF precoder in \eqref{eq:EffSinr21} or the ZF precoder in \eqref{eq:DetEqSinrZf21}. However, the desired signal is affected by phase noise in the same manner as for the RZF and the ZF precoders.
\end{rem}

The SINR results for the RZF, ZF, and MF precoders in  \eqref{eq:EffSinr21}, \eqref{eq:DetEqSinrZf21}, and \eqref{eq:DetEqSinrMf21}, respectively, reduce to the results in \cite{Debbah12}, when phase noise is absent and only AWGN is present, and when the user channels are independent of each other. Alternatively, the results from \cite{Debbah12} can be transformed to the results in \eqref{eq:EffSinr21}, \eqref{eq:DetEqSinrZf21}, and \eqref{eq:DetEqSinrMf21} by changing the quality of the channel estimate from $q_0$ to  $q_0 \mathbb{E}_\phi \absl \T{PN} \absr^2$, where $ \T{PN}$ is as defined in \eqref{eq:SysMod81a}.

%

\subsection{Optimal $\alpha$ for the RZF Precoder in the GO Setup}
\begin{corr}
\label{corr:OptimalAlpha}
Consider an RZF precoded downlink transmission from a BS having $M$ antennas to $K$ single-antenna UEs employing TDD in the presence of oscillator phase noise. Let $\alpha > 0$, $M/K = \beta, \beta \geq 1, q_0  \in  [0,1],$ and let $\mathsf{SINR}_k$ denote the effective SINR achieved by the $k$th UE. Then the optimal $\alpha$ (denoted as $\tilde{\alpha}$) that maximizes $\mathsf{SINR}_k$ in \eqref{eq:EffSinr21}  is given by
\begin{IEEEeqnarray}{rCl} \label{eq:OptAlpha11}
\tilde{\alpha} = \frac{\nsigma{w}{2} +   1-  \mathbb{E}_\phi |\T{PN}|^2 q_0^2    }{  \mathbb{E}_\phi |\T{PN}|^2 q_0^2\beta}
\end{IEEEeqnarray}
where $ \T{PN}$ is as given in \eqref{eq:SysMod81a}.
\begin{proof}
Similar to the steps followed in \cite{Debbah12}, we differentiate \eqref{eq:EffSinr21} with respect to $\alpha$, and set the result equal to zero, which yields $\tilde{\alpha}$ in \eqref{eq:OptAlpha11}.
\end{proof}

\begin{rem}
When perfect channel estimates are available at the BS, i.e., $q_0 = 1$ and $|\T{PN}|^2  = 1$, which holds in the CO setup or when phase noise is absent, then $\tilde{\alpha}  = \nsigma{w}{2}/\beta$, which is similar to the result obtained in \cite{Debbah12}. The regularization parameter $ \tilde{\alpha}$ becomes large when the effective quality of the CSI at the BS is poor, which happens for low values of $q_0$ or $\mathbb{E}_\phi |\T{PN}|^2$ (i.e., severe phase noise at the BS). Under these conditions, the MF precoder becomes optimal. For asymptotically low values of $\nsigma{w}{2}$, the optimal $\alpha$ in \eqref{eq:OptAlpha11} becomes $\limit{\nsigma{w}{2}}{0}  \tilde{\alpha} = \frac{ 1- \mathbb{E}_\phi |\T{PN}|^2 q_0^2    }{ \mathbb{E}_\phi |\T{PN}|^2q_0^2\beta}$, which does not correspond to the ZF precoder, since $ \tilde{\alpha} \neq 0$. However, as $\beta$ becomes large, or when $1- \mathbb{E}_\phi |\T{PN}|^2q_0^2 \rightarrow 0$, the ZF precoder becomes optimal.
\end{rem}
\end{corr}

\section{Simulation Results}
\label{sec:Sim}

\begin{figure}[!t]
\begin{center}
\includegraphics[width = 3.5in, keepaspectratio=true]{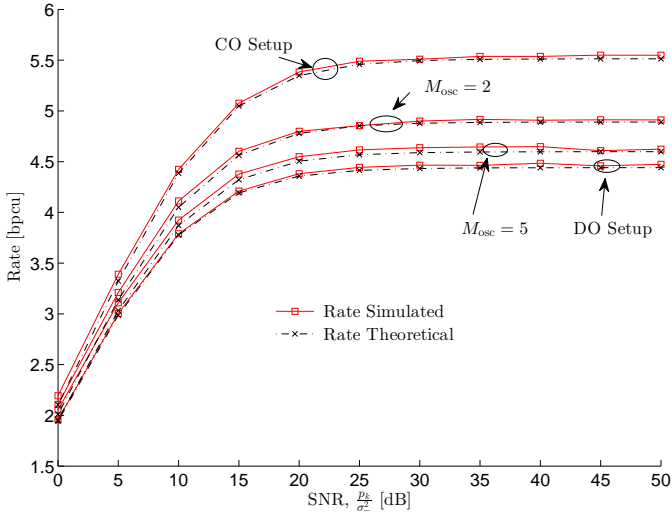}
\caption{$C(\mathsf{SINR}_k)$ for the optimized RZF precoder for $\beta = 5, M = 50,\nsigma{\phi}{} =\nsigma{\varphi}{}   = 6^{\circ}$, and $ q_0 = 0.9$.}
\label{fig:DetEqORZF1}
\end{center}
\end{figure}
In this section, the analytical results for the linear precoders presented in Section \ref{sec:MainResults} are verified by comparing them against the results obtained from MC simulations. Even though the analytical results are derived for $M,K\rightarrow \infty$, in the sequel we observe that these results concur with those from simulations for finite values of $M$ and $K$.

\subsection{Simulation Setup}

We simulate the system model specified in \eqref{eq:SysMod42} using the RZF, ZF, and MF precoders, and numerically evaluate the effective SINR in \eqref{eq:AchRate5}. Then, the achievable rate in the downlink for a given UE is computed using \eqref{eq:AchRate4b} for all values of $M_{\mr{osc}}$, unless stated otherwise. Recall that this evaluation does not account for the differential entropy rates of the phase noise processes at the BS and the UE. We will evaluate the achievable rates for the CO and the DO setups as  $C(\mathsf{SINR}_k) = \min \{\mr{Rate~in~ \eqref{eq:AchRate2}}, \mr{Rate~in~ \eqref{eq:AchRate3}}\}$  \cite{Khanzadi14} when exclusively comparing their performances in Section \ref{sec:COvsDO}. Thereby, the effects of both SINR and the differential entropy rates on the achievable rates of the CO and the DO setups are taken into account.

The system considered consists of a single cell with a BS having $M=50$ antennas. Setting $\beta = 5$, the number of UEs served by the BS is $K = 10$. The channel between a BS antenna and a UE is drawn from an i.i.d. complex Gaussian distribution, i.e., $\h{k}{m} \sim \mathcal{CN}(0,1)$. The coherence time of the channel is set to $\T{c}=100$ data symbol periods, thus resulting in an i.i.d. Rayleigh block-fading channel. The MC simulations are conducted for $10000$ independent channel realizations. The phase noise is simulated as a discrete Wiener process \eqref{eq:SysMod1}, with increment standard deviation $\nsigma{\varphi}{} = \nsigma{\phi}{}= 6^{\circ}$ \cite{Colavolpe05,Rajet13}. Unless stated otherwise, the time elapsed between the training period of the $k$th UE   and the data transmission period is set to $\tau = K=10$ symbol periods. Furthermore, all UEs use the same training power, and for downlink transmission, equal power is allocated by the BS to all  UEs, i.e., $\bP{} = \frac{1}{K}\mathbf{I}_K$. The quality of the channel estimate is set to $q_0 = 0.9$ \cite{Debbah12}. This is reasonable given that the UEs can choose to transmit at power levels ${\pu{k}}$ such that the desired $q_0$ is attained.
\begin{figure}[!t]
\begin{center}
\includegraphics[width = 3.5in, keepaspectratio=true]{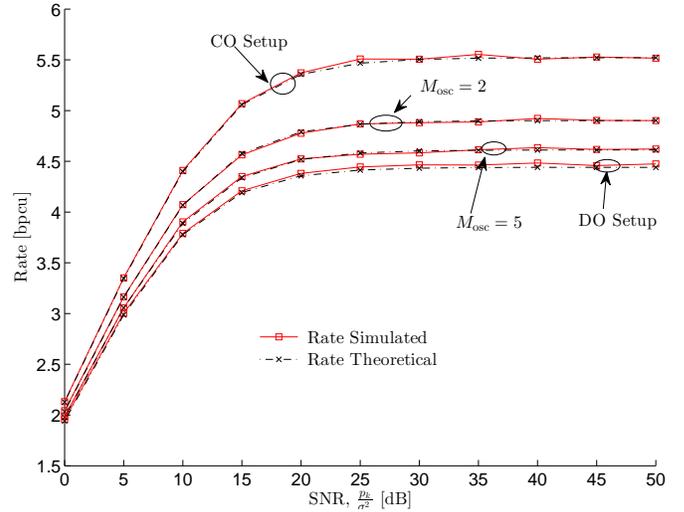}
\caption{$C(\mathsf{SINR}_k)$ for the ZF precoder for $\beta = 5, M = 50,\nsigma{\phi}{} =\nsigma{\varphi}{}  = 6^{\circ}$, and $ q_0 = 0.9$.}
\label{fig:DetEqZF1}
\end{center}
\end{figure}

\subsection{Verification of Analytical Results}

Fig. \ref{fig:DetEqORZF1} compares the rate achieved with  RZF precoded transmission from the BS to the UEs for different SNR values, where the SNR at the $k$th  UE is defined as $\frac{p_k}{\nsigma{w}{2}}$. The value of $\alpha$ used in this simulation is evaluated using \eqref{eq:OptAlpha11}. We see that the achievable rate \eqref{eq:AchRate4b} for the RZF precoder based on the effective SINR in \eqref{eq:EffSinr21} concurs with the rate achieved in simulations for all the values of $M_{\mr{osc}}$ considered. Furthermore,  the SINR of the RZF precoder decreases with increasing $M_{\mr{osc}}$. Therefore, in terms of SINR degradation due to channel aging caused by phase noise, the CO setup  is more robust than the GO setup when $M_{\mr{osc}}\geq 2$. Also, the performance of the GO setup for $M_{\mr{osc}} = 5$ is close to that of the DO setup. 

In Fig. \ref{fig:DetEqZF1}, we compare the rate achieved when a ZF precoder is used for transmission from the BS to the UEs. We observe that our analytical results for the achievable rate using the SINR in \eqref{eq:DetEqSinrZf21} matches with those obtained by simulations for all the  considered values of $M_{\mr{osc}}$. As in the case of the RZF precoder, the SINR performance of the ZF precoder degrades with increasing $M_{\mr{osc}}$.

\begin{figure}[!t]
\begin{center}
\includegraphics[width = 3.5in, keepaspectratio=true]{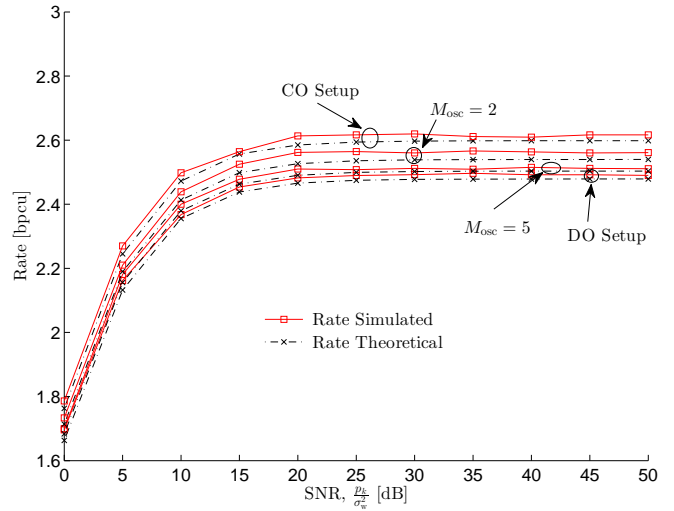}
\caption{$C(\mathsf{SINR}_k)$ for the MF precoder for $\beta = 5, M = 50,\nsigma{\phi}{} =\nsigma{\varphi}{}  = 6^{\circ}$, and $  q_0 = 0.9$.}
\label{fig:DetEqMF1}
\end{center}
\end{figure}

Fig. \ref{fig:DetEqMF1} compares the performance achieved using the MF precoder in simulations with the performance that is obtained using the SINR given  in \eqref{eq:DetEqSinrMf21}. The SINR performance of the CO setup is better than that of the GO setup when $M_{\mr{osc}}\geq 2$. Also, the gap in performance between the CO setup and the DO setup is smaller when compared to the case when the RZF and the ZF precoders are used. This is because  the interference power for the MF precoder does not depend on $\T{PN}$ or $q_0$ \eqref{eq:DetEqSinrMf21}. 

\subsection{Optimal $\alpha$ and Precoder Performance Comparison}
In  Fig. \ref{fig:OptAlph2}, $\tilde{\alpha}$, which is numerically determined, is compared with the analytical result given in \eqref{eq:OptAlpha11}. In our numerical simulations, we exhaustively search for $\tilde{\alpha}$, which maximizes the SINR in \eqref{eq:EffSinr21}, when  $\beta = 5, M = 50,\nsigma{\phi}{} =\nsigma{\varphi}{}  = 6^{\circ}$, and $q_0 = 0.9$. Clearly, the values of $\tilde{\alpha}$ obtained from \eqref{eq:OptAlpha11} agree with those obtained from simulations. Moreover, at high SNR, we observe that the optimal linear precoder is not the ZF precoder. Furthermore, since the effective CSI quality in the GO setup when $M_{\mr{osc}}\geq 2$, is lower than that in the CO setup, $\tilde{\alpha}$ is relatively larger for the former case.

\begin{figure}[!t]
\begin{center}
\includegraphics[width = 3.5in, keepaspectratio=true]{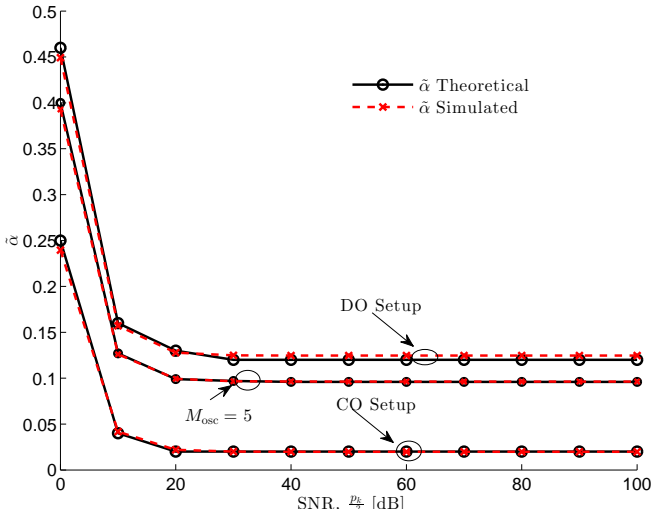}
\caption{$\tilde{\alpha}$ for $\beta = 5, M = 200,\nsigma{\phi}{2} =\nsigma{\varphi}{2}  = 6^{\circ}$, and $  q_0 = 0.9$.}
\label{fig:OptAlph2}
\end{center}
\end{figure}

\begin{figure}[!t]
\begin{center}
\includegraphics[width = 3.5in, keepaspectratio=true]{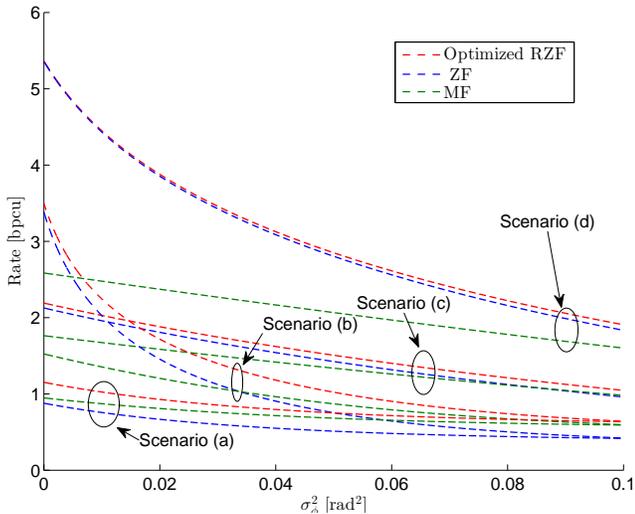}
\caption{$C(\mathsf{SINR}_k)$ of the different precoders in the GO setup for $M=50,M_{\mr{osc}} = 5, q_0 = 0.9$, and different $\nsigma{\phi}{2}$ values. The operating scenarios considered are the following: Scenario (a): $\mr{SNR}=0$ dB, $\beta = 2$. Scenario (b): $\mr{SNR}=20$ dB, $\beta = 2$.  Scenario (c): $\mr{SNR}=0$ dB, $\beta = 5$.  Scenario (d): $\mr{SNR}=20$dB, $\beta = 5$.}
\label{fig:RateDiffVarDO}
\end{center}
\end{figure}

We compare $C(\mathsf{SINR}_k)$ of the considered precoders  for different $\nsigma{\phi}{2} $ values and operating scenarios  in Fig. \ref{fig:RateDiffVarDO}. We set $M=50, M_{\mr{osc}}=5,$ and $q_0 = 0.9$ for all scenarios. In Scenario (a), we set $\beta = 2$ and $\mr{SNR}=0$ dB. In this scenario, the performance of the MF precoder is consistently better than that of the ZF precoder, and approaches the performance of the optimized RZF precoder as $\nsigma{\phi}{2} $ increases. As the SNR is increased to $20$ dB in Scenario (b),  the ZF precoder performs better than the MF precoder for low values of $\nsigma{\phi}{2}$. But for large  $\nsigma{\varphi}{2} $ values, the order is reversed.

In Scenario (c), we set the SNR to $0$ dB,  and increase $\beta$ to 5, thereby reducing the interference. In this scenario, the ZF precoder is seen to perform better than the MF precoder, except when the phase noise is severe. Finally, in Scenario (d), the SNR is increased to 20 dB, and the ZF precoder significantly outperforms the MF precoder for all $\nsigma{\varphi}{2} $ values considered. In summary, we conclude that the relative performance of the MF and ZF precoders depends on the operating scenario, which depends on the values of $M_{\mr{osc}}$, $M$, $q_0$, $\nsigma{\varphi}{2}$, SNR, and $\beta$.


\begin{figure}[!t]
\begin{center}
\includegraphics[width = 3.2in, keepaspectratio=true]{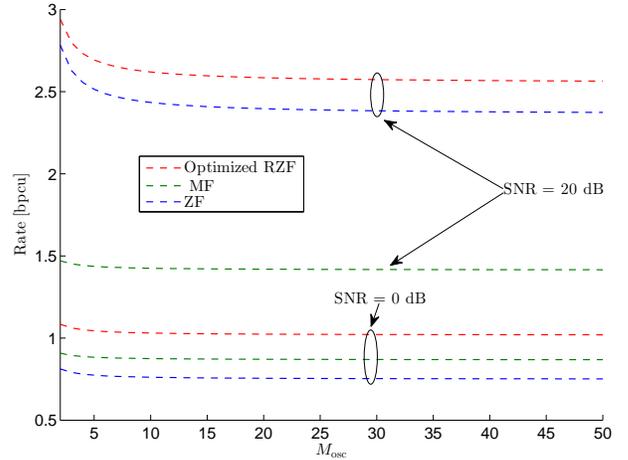}
\caption{$C(\mathsf{SINR}_k)$ of the different precoders in the GO setup, where $1 \leq M_{\mr{osc}}\leq M, M/M_{\mr{osc}} \in \mathbb{Z}^+$, for $M=50, \beta = 2, q_0 = 0.9$, $\nsigma{\phi}{}=\nsigma{\varphi}{} = 0.06^\circ$, $\tau = \T{c} = 0.25~\tr{ms}$  ($10^4$ symbol periods) \cite{Lte}, and different SNR values.  }
\label{fig:RateDiffMOsc}
\end{center}
\end{figure}

\subsection{Example Based on LTE Specifications}

We analyze an example based on long-term evolution (LTE) system specifications \cite{Marzetta10,Lte}, where we account for practical values of $\T{c}$, symbol time $\T{s}$, bandwidth $BW$, $\tau$, center frequency $f_{\mr{c}}$, doppler spread $f_{\mr{d}}$, $\nsigma{\phi}{}$, and $\nsigma{\varphi}{}$. We choose $\T{s} = 0.032\mu s$, $BW = 20$ MHz, $f_{\mr{c}} = 2$ GHz, and $f_{\mr{d}} = 1000$ Hz arising from a relative velocity of $500$ km/h between the BS and the UEs. Letting $\T{c} = 1/4f_{\mr{d}}$, the coherence time is $\T{c} = 0.25$ ms. We also consider that the time elapsed between the training and the data transmission phase is equal to the coherence time of the channel, i.e., $\tau=\T{c}=0.25$ ms \cite[pg. 99]{Lte}.

Next, we compute  $\nsigma{\phi}{2}$ and $\nsigma{\varphi}{2}$ based on a Si CMOS oscillator technology in \cite{AnalogDevices}. Specifically, we consider an oscillator, whose offset level at $90$ MHz is $-156$ dBc/Hz. This renders $\nsigma{\phi}{2}=\nsigma{\varphi}{2} = 10^{-6}$ rad$^2$, or  $\nsigma{\phi}{}=\nsigma{\varphi}{} = 0.06 {}^\circ$ using \cite[Eq. (4)]{Rajet14_1}, implying that high-quality oscillators are used at the BS and the UE. In Fig. \ref{fig:RateDiffMOsc}, we plot the performance of the precoders versus $M_{\mr{osc}}$ for $\tau=\T{c}=0.25$ ms ($10^4$ symbol periods), $q_0 = 0.9$, $\beta = 2$, and $M=50$ for different SNR values. At $\mr{SNR}=20$ dB, we observe that the performance of the RZF and ZF precoders decreases by around $0.3$ bits per channel use (bpcu), as $M_{\mr{osc}}$ increases, and for $M_{\mr{osc}}>10$ oscillators at the BS, the additional degradation in performance becomes negligible. One the other hand, the degradation in the performance of the MF precoder is negligible for all values of $M_{\mr{osc}}$ considered. Furthermore, at $\mr{SNR}=0$ dB, the performance degradation for all considered precoders is negligible as $M_{\mr{osc}}$ is increased from 1 to $M$.


\subsection{Rate Comparisons for CO and DO Setups}
\label{sec:COvsDO}

We compare the performance of the CO setup with that of the DO setup by computing the achievable rate as   $C(\mathsf{SINR}_k) = \min \{\mr{Rate~in~ \eqref{eq:AchRate2}}, \mr{Rate~in~ \eqref{eq:AchRate3}}\}$. 
We set $q_0 = 0.9$, $\nsigma{\phi}{}=\nsigma{\varphi}{} = 6^\circ$, and $\tau = K = 25$, and analyze the rate performance of the RZF and MF precoders for different values of $\beta$ in Figs. \ref{fig:RzfCOvsDO} and \ref{fig:MfCOvsDO}.

We first consider the optimized RZF precoder in Fig. \ref{fig:RzfCOvsDO}. For $\mr{SNR}=40$ dB, the performance of the CO setup is consistently better than that of the DO setup as the SINR used in  \eqref{eq:AchRate2} and \eqref{eq:AchRate3} is large, and exclusively determines the achievable rate. In particular, when $\beta$ is small, the rate in \eqref{eq:AchRate2} is a tighter upper bound, and the achievable rate is determined by the SINR term, which is relatively larger for the CO setup. As $\beta$ increases, the rate in \eqref{eq:AchRate3} becomes a tighter upper bound, and the SINR term is much larger than the differential entropy term. Hence, the CO setup still performs better. Now, consider the low SNR scenario ($\mr{SNR}=0$ dB). Here, once again, for small $\beta$, the rate in \eqref{eq:AchRate2} is a tighter upper bound than that in \eqref{eq:AchRate3}, and depends on the SINR alone. Therefore, the CO setup performs better. However, as $\beta$ increases, the rate in \eqref{eq:AchRate3} becomes a tighter upper bound, and the differential entropy term becomes significant compared to the SINR term. In particular, the DO setup has a lower differential entropy rate as it is only impaired by the phase noise at the UE. Consequently, the DO setup performs better.

A similar performance order is observed  in Fig. \ref{fig:MfCOvsDO} for the MF precoder when $\mr{SNR}=0$ dB. For $\mr{SNR}=40$ dB and low $\beta$ values, as before, the achievable rate depends on the SINR alone as the rate in \eqref{eq:AchRate2} is a tighter upper bound. However, as $\beta$ increases, the rate in \eqref{eq:AchRate3} becomes a tighter upper bound. Unlike in the case of the RZF precoder, now  the differential entropy term is significant compared to the SINR term in \eqref{eq:AchRate3}. This is because the SINR of the MF precoder is significantly lower than that of the RZF precoder. Moreover, the difference in the SINR for the CO and the DO setups is not as significant as when RZF precoding is used, since the interference power for the MF precoder is the same for both the setups \eqref{eq:DetEqSinrMf21}. Therefore, the DO setup performs better than the CO setup.
\begin{figure}[!t]
\begin{center}
\includegraphics[width = 3.25in, keepaspectratio=true]{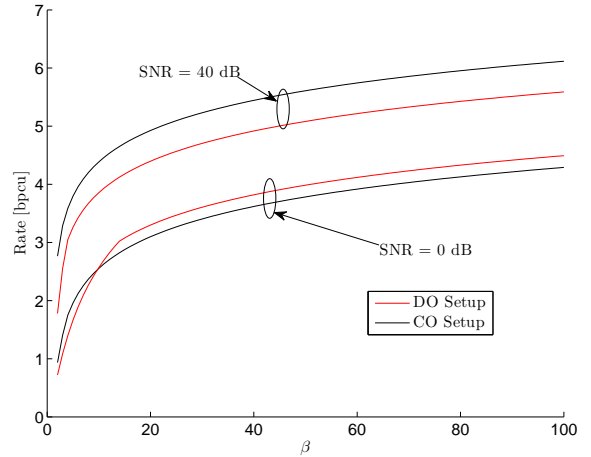}
\caption{$C(\mathsf{SINR}_k) = \min \{\mr{Rate~in~ \eqref{eq:AchRate2}}, \mr{Rate~in~ \eqref{eq:AchRate3}}\}$ of the optimized RZF precoder for the CO and DO setups, where $q_0 = 0.9$, $\nsigma{\phi}{}=\nsigma{\varphi}{} = 6^\circ$, and $\tau = K =25$.  }
\label{fig:RzfCOvsDO}
\end{center}
\end{figure}

\begin{figure}[!t]
\begin{center}
\includegraphics[width = 3.25in, keepaspectratio=true]{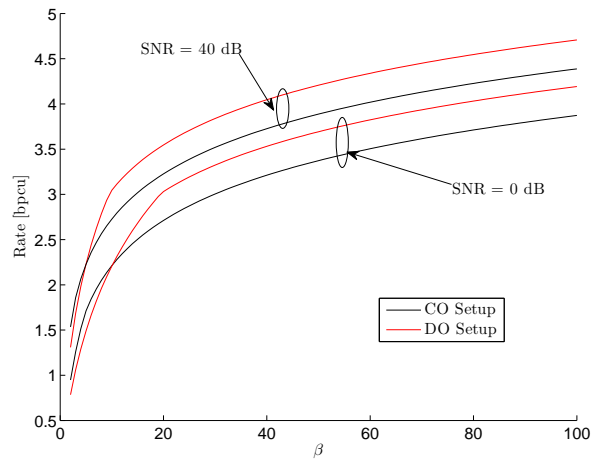}
\caption{$C(\mathsf{SINR}_k) = \min \{\mr{Rate~in~ \eqref{eq:AchRate2}}, \mr{Rate~in~ \eqref{eq:AchRate3}}\}$ of the MF precoder for the CO and DO setups, where $q_0 = 0.9$, $\nsigma{\phi}{}=\nsigma{\varphi}{} = 6^\circ$, and $\tau = K =25$.  }
\label{fig:MfCOvsDO}
\end{center}
\end{figure}

\section{Conclusions}
\label{sec:Conc}

In this work, we derived the effective SINR of the RZF, ZF, and MF precoders  for the GO setup  in the presence of phase noise. We showed that the effect of phase noise on the SINR can be expressed as an effective reduction in the CSI quality available at the BS. Importantly,  the SINR performance of all considered precoders degrades as the number of oscillators, $M_{\mr{osc}}$, increases. This is because as  $M_{\mr{osc}}$ increases, the desired signal power decreases, and the interference power increases. However, for the MF precoder, the interference power is almost independent of the phase noise. Furthermore, we showed that the variance of the  random phase variations caused by the BS oscillators reduces with increasing $M_{\mr{osc}}$. By simulations, we demonstrated that the SINR approximations obtained are tight, and agree with those obtained from simulations with remarkable accuracy for interesting, and practical values of $M$ and $K$.

We showed that the optimized RZF has superior SINR performance compared to the ZF and the MF precoders for all considered scenarios. Moreover, the ZF precoder performs better than the MF precoder when the CSI available at the BS is reliable, the phase noise at the BS is not severe, and $\beta$ is large. However, for low SNR,  severe phase noise at the BS, and small $\beta$, the MF performs better. Finally, we observed that for all considered precoders, the CO setup has a higher achievable rate than the DO setup when $\beta$ is small, while the DO setup outperforms the CO setup when the SNR at the UE is low and $\beta$ is large.


\section*{APPENDIX A}
\section*{Important Results from Literature}
\label{sec:app_a}

\begin{definition}[{Stieltjes Transform \cite[Section 2.2]{VerduRMT}}]
 \label{definition:stieltjes}
Let $X$ be a real-valued RV with distribution $F$. The the Stieltjes transform $m(z)$ of $F$, for $z\in \mathbb{C}$ such that $\Im \{z\} > 0$,  is defined as
\begin{IEEEeqnarray}{rCl} \label{eq:Stiel}
m(z) = \mathbb{E}\squarel  \frac{1}{X - z} \squarer =  \int_{-\infty}^{\infty}  \frac{1}{\lambda - z} dF(\lambda).
\end{IEEEeqnarray}
\end{definition}

\begin{lem}
 \label{lem:stieltjes2}
\emph{Stieltjes Transform of the Marchenko-Pastur Law  \cite[Eqs. (1.12, 2.43)]{VerduRMT}:} Let $\bH{} \in \mathbb{C}^{K \times M}$, with entries that are zero-mean i.i.d. RVs  with variance $1/M$.  Then the empirical distribution of the eigenvalues of $\bH{}^{\mr{H}}\bH{}$ converges to the Marchenko-Pastur law almost surely as $M,K \rightarrow \infty$, with $ M/K = \beta$. Furthermore, the Stieltjes transform $m(z)$, with complex argument $z$ such that $\Im \{z\} > 0$,  of the Marchenko-Pastur density law is defined as
\begin{IEEEeqnarray}{rCl} \label{eq:Stiel2}
m(z) = \frac{ 1 - z\beta - \beta \pm \sqrt{\beta^2z^2 - 2(\beta + 1)z\beta + (1 - \beta)^2} }{2z\beta} \IEEEeqnarraynumspace
\end{IEEEeqnarray}
\end{lem}

\begin{lem}[{Matrix Inversion Lemma  \cite[Eq. (2.2)]{Silverstein95}}]
 \label{lem:MatInvLem}
For an invertible matrix $\bU{} \in \mathbb{C}^{M\times M}$, $\bh{} \in \mathbb{C}^{M\times 1}$, and $q \in \mathbb{C}$, where $\bU{} + q\bh{}\bh{}^{\mr{H}}$ is invertible,
\begin{IEEEeqnarray}{rCl} \label{eq:MIL}
\bh{}^{\mr{H}}(\bU{} + q\bh{}\bh{}^{\mr{H}})^{\minus 1} = \frac{\bh{}^{\mr{H}} \bU{}^{\minus 1}}{1 + q\bh{}^{\mr{H}}\bU{}^{\minus 1}\bh{} },
\end{IEEEeqnarray}
since $\bh{}^{\mr{H}} \bU{}^{\minus 1}(\bU{} + q\bh{}\bh{}^{\mr{H}}) =  (1 + q\bh{}^{\mr{H}}\bU{}^{\minus 1}\bh{})\bh{}^{\mr{H}}$.
\end{lem}

\begin{lem}[{Resolvent Identity \cite[Lemma 2]{Debbah12}}]
 \label{lem:ResIdn}
Given two invertible matrices $\bU{}$ and $\bV{}$ of size $M \times M$,
\begin{IEEEeqnarray}{rCl} \label{eq:ResIdn}
\bU{}^{\minus 1} - \bV{}^{\minus 1} = -\bU{}^{\minus 1}(\bU{} - \bV{})\bV{}^{\minus 1}
\end{IEEEeqnarray}
holds.
\end{lem}

\begin{lem}[{Trace Lemma \cite[Lemma 4]{Debbah12}}] 
 \label{lem:TrLem}
Let $\bx{}, \bw{} \sim \Ccal \Ncal (0,\frac{1}{M}\mathbf{I}_{M})$ be mutually independent vectors of length $M$, and also independent of $\bA{} \in \mathbb{C}^{M \times M}$, which has a uniformly bounded spectral norm for all $M$. Then
\begin{IEEEeqnarray}{rCl} \label{eq:TrLem}
\bx{}^{\mr{H}} \bA{} \bx{} - \frac{1}{M} \mr{tr} \bA{} \overset{M \rightarrow \infty}{\longrightarrow} 0, \bx{}^{\mr{H}} \bA{} \bw{}  \overset{M \rightarrow \infty}{\longrightarrow} 0.\IEEEeqnarraynumspace
\end{IEEEeqnarray}
\end{lem}

\begin{lem}
 \label{lem:stiel}
Let $\bH{} \in \mathbb{C}^{K \times M}$, $M,K \rightarrow \infty$ with $ M/K = \beta$, whose entries are zero-mean i.i.d. Gaussian RVs with variance $1/M$, and define $\bA{} \triangleq \frac{1}{M}\bH{}^{\mr{H}}\bH{} + \alpha \mathbf{I}_M$. Then
\begin{IEEEeqnarray}{rCl} \label{eq:R1PTrLem}
\frac{1}{M} {\mr{tr} \bA{}^{\minus 1} } -  m(-\alpha)  \longrightarrow  0, \frac{1}{M} {\mr{tr} \bA{}^{-2} } -  m'(-\alpha) \longrightarrow  0 \IEEEeqnarraynumspace
\end{IEEEeqnarray}
\end{lem}

\begin{lem}[{Rank-1 Perturbation Lemma \cite[Lemma 2.6]{Silverstein95}}]
 \label{lem:R1PLem}
Let $\zeta>0$, $\bU{}, ~\tr{and}~ \bA{} \in \mathbb{C}^{M \times M}$ with $\bU{}$ being nonnegative Hermitian, $\bh{} \in \mathbb{C}^{M\times 1}$, and $q \in \mathbb{R}$. Then,
\begin{IEEEeqnarray}{rCl} \label{eq:R1PLem}
|\mr{tr} \bA{} \squarel  (\bU{} + \zeta \mathbf{I}_M + q\bh{}\bh{}^{\mr{H}} )^{\minus 1} - (\bU{} + \zeta \mathbf{I}_M)^{\minus 1}  \squarer |  \leq \frac{|| \bA{} ||}{\zeta}. \IEEEeqnarraynumspace
\end{IEEEeqnarray}
\end{lem}

\begin{lem}[{ \cite[Lemma 6]{Debbah12}}]
 \label{lem:R1PTrLem}
Let $\bU{}, \bA{} \in \mathbb{C}^{M \times M}$ with $\bU{}$ being nonnegative Hermitian, $\bh{} \in \mathbb{C}^{M}$ and $q \in \mathbb{R}$, then,
\begin{IEEEeqnarray}{rCl} \label{eq:R1PTrLem}
\frac{1}{M} {\mr{tr} \bA{} \bU{}^{\minus 1} } - \frac{1}{M} {\mr{tr} \bA{} (\bU{} +  q\bh{}\bh{}^{\mr{H}})^{\minus 1} }   \overset{M \rightarrow \infty}\longrightarrow  0 \IEEEeqnarraynumspace
\end{IEEEeqnarray}
\end{lem}

\begin{lem}[{ \cite[Lemma 7]{Debbah12}}]
\label{lem:Lemm7Debb}
Consider $\bU{}, \bA{} \in \mathbb{C}^{M \times M}$ with uniformly bounded spectral norms for all $M$ with $\bA{}$ being invertible. Furthermore, let $\bx{}, \bw{} \sim \Ccal \Ncal (0,\frac{1}{M}\mathbf{I}_{M})$ be mutually independent vectors of length $M$, and also independent of $\bU, \bA{}$. Define $q_0, q_1, q_2 \in \mathbb{R}^+$ such that $q_0q_1  = q_2^2$, $q_0 + q_1 = 1$ and $t_1 \triangleq \frac{1}{M} \mr{tr} \bA{}^{\minus 1}, t_2 \triangleq \frac{1}{M} \mr{tr} \bU{} \bA{}^{\minus 1}$. Then,
\begin{IEEEeqnarray}{rCl} \label{eq:Lemm7Debb}
&& \bx{}^{\mr{H}} \bU{} \parl \bA{} + q_0 \bx{}\bx{}^{\mr{H}} + q_1  \bw{}\bw{}^{\mr{H}} + q_2\bx{}\bw{}^{\mr{H}}  + q_2\bw{}\bx{}^{\mr{H}} \parr^{\minus 1} \bx{} \nonumber\\
&& - \frac{t_2(1 + q_1 t_1)}{1 + t_1} \overset{M \rightarrow \infty}{\longrightarrow} 0 \\
&& \bw{}^{\mr{H}} \bU{} \parl \bA{} + q_0 \bx{}\bx{}^{\mr{H}} + q_1  \bw{}\bw{}^{\mr{H}} + q_2\bx{}\bw{}^{\mr{H}}  + q_2\bw{}\bx{}^{\mr{H}} \parr^{\minus 1} \bw{} \nonumber\\
&& - \frac{t_2(1 + q_0 t_1)}{1 + t_1} \overset{M \rightarrow \infty}{\longrightarrow} 0 \\
&& \bx{}^{\mr{H}} \bU{} \parl \bA{} + q_0 \bx{}\bx{}^{\mr{H}} + q_1  \bw{}\bw{}^{\mr{H}} + q_2\bx{}\bw{}^{\mr{H}}  + q_2\bw{}\bx{}^{\mr{H}} \parr^{\minus 1} \bw{} \nonumber\\
&& - \frac{-q_2 t_1 t_2}{1 + t_1} \overset{M \rightarrow \infty}{\longrightarrow} 0 \\
&& \bw{}^{\mr{H}} \bU{} \parl \bA{} + q_0 \bx{}\bx{}^{\mr{H}} + q_1  \bw{}\bw{}^{\mr{H}} + q_2\bx{}\bw{}^{\mr{H}}  + q_2\bw{}\bx{}^{\mr{H}} \parr^{\minus 1} \bx{} \nonumber\\
&& - \frac{-q_2 t_1 t_2}{1 + t_1} \overset{M \rightarrow \infty}{\longrightarrow} 0.
\end{IEEEeqnarray}
\end{lem}

\begin{lem}
 \label{lem:FrPr}
Let $\bU{}, \bV{} \in \mathbb{C}^{M \times M}$ be freely independent random matrices  \cite[Page~207]{Tao} with uniformly bounded spectral norm for all $M$. Further, let all the moments of the entries of $\bU{}, \bV{}$ be finite. Then,
\begin{IEEEeqnarray}{rCl} \label{eq:TrLem}
\frac{1}{M}\mr{tr} \bU{} \bV{} -   \frac{1}{M}\mr{tr} \bU{} \frac{1}{M}\mr{tr} \bV{} \overset{M \rightarrow \infty}{\longrightarrow} 0.
\end{IEEEeqnarray}
\end{lem}

\section*{APPENDIX B}
\section*{Extensions to Existing Lemmas in \cite{Debbah12}}
\label{sec:app_b}

First, we provide Lemma \ref{lem:NewLemDebb}, which is an extension to Lemma \ref{lem:Lemm7Debb} \cite[Lemma 7]{Debbah12}. Lemma \ref{lem:NewLemDebb} is then used to the derive the effective SINR in \eqref{eq:AchRate5}.

\begin{lem} [{Extensions to \cite[Lemma 7]{Debbah12}}]
\label{lem:NewLemDebb}
Consider $\bU{}$, $\bA{}$, $\bN{} \in \mathbb{C}^{M \times M}$ with uniformly bounded spectral norms for all $M$, such that $\bA{}$ is invertible, and $\bN{}$ is unitary, $\bN{}\bN{}^{\mr{H}} = \mathbf{I}_M$. Furthermore, assume that $\bN{}$ is freely independent of $\bU{}, \bA{}$ \cite[Page~207]{Tao}.  Let $\bx{}, \bw{} \sim \Ccal \Ncal (0,\frac{1}{M}\mathbf{I}_{M})$ be mutually independent vectors of length $M$, and also independent of $\bU, \bA{}$. Define $q_0, q_1, q_2 \in \mathbb{R}^+$ such that $q_0q_1  = q_2^2$, $q_0 + q_1 = 1$ and $t_1 \triangleq \frac{1}{M} \mr{tr} \bA{}^{\minus 1}, t_2 \triangleq \frac{1}{M} \mr{tr} \bU{} \bA{}^{\minus 1}$. Then,
\begin{IEEEeqnarray}{rCl}
\label{eq:NewLemmDebb01}
&& \bx{}^{\mr{H}} \bN{}\bU{} \parl \bA{} + q_0 \bx{}\bx{}^{\mr{H}} + q_1  \bw{}\bw{}^{\mr{H}} + q_2\bx{}\bw{}^{\mr{H}}  + q_2\bw{}\bx{}^{\mr{H}} \parr^{\minus 1} \bN{}^{\mr{H}} \bx{} \nonumber\\
&& - \parl t_2 -  \frac{q_0 t_1 t_2 }{1 + t_1}\absl\frac{ \mr{tr} \curll  \bN{} \curlr }{M} \absr^2  \parr \overset{M \rightarrow \infty}{\longrightarrow} 0 \\
\label{eq:NewLemmDebb02}
&& \bx{}^{\mr{H}} \bU{} \parl \bA{} + q_0 \bx{}\bx{}^{\mr{H}} + q_1  \bw{}\bw{}^{\mr{H}} + q_2\bx{}\bw{}^{\mr{H}}  + q_2\bw{}\bx{}^{\mr{H}} \parr^{\minus 1} \bN{}^{\mr{H}} \bx{} \nonumber\\
&& - \frac{t_2 (1 + q_1 t_1) }{1 + t_1} \frac{ \mr{tr} \curll  \bN{}^{\mr{H}} \curlr }{M} \overset{M \rightarrow \infty}{\longrightarrow} 0 \\
\label{eq:NewLemmDebb03}
&& \bw{}^{\mr{H}} \bU{}\parl \bA{} + q_0 \bx{}\bx{}^{\mr{H}} + q_1  \bw{}\bw{}^{\mr{H}} + q_2\bx{}\bw{}^{\mr{H}}  + q_2\bw{}\bx{}^{\mr{H}} \parr^{\minus 1}\bN{}^{\mr{H}}\bx{}  \nonumber\\
&& - \frac{-q_2 t_1 t_2 }{1 + t_1}  \frac{ \mr{tr} \curll  \bN{}^{\mr{H}} \curlr }{M} \overset{M \rightarrow \infty}{\longrightarrow} 0
\end{IEEEeqnarray}

\begin{proof}
Since $\bN{}$ is freely independent of $\bU{}, \bA{}$,  from Lemma \ref{lem:FrPr} we have
\begin{IEEEeqnarray}{rCl}
\label{eq:NewLemmDebb11}
&& \frac{1}{M} \mr{tr} \bN{} \bA{}^{\minus 1} - \frac{1}{M} \mr{tr} \bN{}  \underbrace{\frac{1}{M} \mr{tr} \bA{}^{\minus 1}}_{\triangleq t_1} \overset{M \rightarrow \infty}{\longrightarrow} 0  \\
 \label{eq:NewLemmDebb12}
&& \frac{1}{M} \mr{tr} \bN{} \bU{} \bA{}^{\minus 1} -  \frac{1}{M} \mr{tr} \bN{} \underbrace{\frac{1}{M} \mr{tr} \bU{} \bA{}^{\minus 1}}_{\triangleq t_2} \overset{M \rightarrow \infty}{\longrightarrow} 0.
\end{IEEEeqnarray}

Furthermore,
\begin{IEEEeqnarray}{rCl}
 \label{eq:NewLemmDebb21}
&& \bx{}^{\mr{H}}\bN{}\bU{}\bV{}\bx{} - \frac{ \mr{tr}   \bN{}  }{M} \frac{t_2(1 + q_1 t_1)}{1 + t_1} \overset{M \rightarrow \infty}{\longrightarrow} 0 \\
 \label{eq:NewLemmDebb22}
&& \bx{}^{\mr{H}}\bA{}^{\minus 1}\bN{}^{\mr{H}}\bx{} - \frac{ \mr{tr}   \bN{}^{\mr{H}}  }{M} t_1 \overset{M \rightarrow \infty}{\longrightarrow} 0 \\
 \label{eq:NewLemmDebb222}
&& \bx{}^{\mr{H}} \bU{} \bA{}^{\minus 1}\bN{}^{\mr{H}}\bx{} - \frac{ \mr{tr}   \bN{}^{\mr{H}}  }{M} t_2 \overset{M \rightarrow \infty}{\longrightarrow} 0 \\
 \label{eq:NewLemmDebb23}
&& \bx{}^{\mr{H}}\bN{}\bU{}\bV{}\bw{} -  \frac{-q_2 t_1 t_2 }{1 + t_1}\frac{ \mr{tr}   \bN{}  }{M} \overset{M \rightarrow \infty}{\longrightarrow} 0 \\
 \label{eq:NewLemmDebb24}
&& \bx{}^{\mr{H}}\bN{}\bU{}\bA{}^{\minus 1}\bN{}^{\mr{H}}\bx{} - t_2 \overset{M \rightarrow \infty}{\longrightarrow}  0 \\
 \label{eq:NewLemmDebb25}
&& \bw{}^{\mr{H}}\bA{}^{\minus 1}\bN{}^{\mr{H}}\bx{} \overset{M \rightarrow \infty}{\longrightarrow} 0.
\end{IEEEeqnarray}
In order to obtain \eqref{eq:NewLemmDebb21} and \eqref{eq:NewLemmDebb23}, the result in \eqref{eq:NewLemmDebb12} along with Lemma \ref{lem:Lemm7Debb} is applied. The results in \eqref{eq:NewLemmDebb22},  \eqref{eq:NewLemmDebb222}, \eqref{eq:NewLemmDebb24}, and \eqref{eq:NewLemmDebb25} are obtained by using \eqref{eq:NewLemmDebb11}, \eqref{eq:NewLemmDebb12}, and Lemma \ref{lem:TrLem}. Now, define $\bV{} \triangleq (\bA{}  + q_0\bx{}\bx{}^{\mr{H}}  + q_1\bw{}\bw{}^{\mr{H}} + q_2\bw{}\bx{}^{\mr{H}} + q_2\bx{}\bw{}^{\mr{H}}  )^{\minus 1},$ then
\begin{IEEEeqnarray}{rCl}
\label{eq:NewLemmDebb31}
&& \bx{}^{\mr{H}}\bN{}\bU{}\bV{}\bN{}^{\mr{H}}\bx{}  \nonumber\\
&& =  \bx{}^{\mr{H}}\bN{}\bU{}\bA{}^{\minus 1}\bN{}^{\mr{H}}\bx{} - \bx{}^{\mr{H}}\bN{}\bU{}(\bV{}^{\minus 1} - \bA{})\bA{}^{\minus 1}\bN{}^{\mr{H}}\bx{} \\
\label{eq:NewLemmDebb32}
&& =  \bx{}^{\mr{H}}\bN{}\bU{}\bA{}^{\minus 1}\bN{}^{\mr{H}}\bx{} - \bx{}^{\mr{H}}\bN{}\bU{}\bV{} \left(q_0\bx{}\bx{}^{\mr{H}}   + q_1\bw{}\bw{}^{\mr{H}} + q_2\bw{}\bx{}^{\mr{H}} \right. \nonumber\\ && \left. ~+ q_2\bx{}\bw{}^{\mr{H}}\right)\bA{}^{\minus 1}\bN{}^{\mr{H}}\bx{}\\
\label{eq:NewLemmDebb33}
&& = -q_0 \bx{}^{\mr{H}}\bN{}\bU{}\bV{}\bx{}\bx{}^{\mr{H}}\bA{}^{\minus 1}\bN{}^{\mr{H}}\bx{}  - q_1\bx{}^{\mr{H}}\bN{}\bU{}\bV{}\bw{}\bw{}^{\mr{H}}\bA{}^{\minus 1}\bN{}^{\mr{H}}\bx{} \nonumber\\
\label{eq:NewLemmDebb34}
&& ~- q_2\bx{}^{\mr{H}}\bN{}\bU{}\bV{}\bw{}\bx{}^{\mr{H}}\bA{}^{\minus 1}\bN{}^{\mr{H}}\bx{} - q_2\bx{}^{\mr{H}}\bN{}\bU{}\bV{}\bx{}\bw{}^{\mr{H}}\bA{}^{\minus 1}\bN{}^{\mr{H}}\bx{} \IEEEeqnarraynumspace\\
\label{eq:NewLemmDebb35}
&& = t_2 -  \frac{q_0 t_1 t_2(1 + q_1 t_1)}{1 + t_1}\absl \frac{ \mr{tr} \curll  \bN{} \curlr }{M} \absr^2  +   \frac{q_2^2 t_1^2 t_2 }{1 + t_1} \absl \frac{ \mr{tr} \curll  \bN{} \curlr }{M}\absr^2 \\\label{eq:NewLemmDebb36}
&& = t_2 -  \frac{q_0 t_1 t_2 }{1 + t_1}\absl \frac{ \mr{tr} \curll  \bN{} \curlr }{M}\absr^2.
\end{IEEEeqnarray}
In \eqref{eq:NewLemmDebb31}, the resolvent identity in Lemma \ref{lem:ResIdn} is used. Upon simplifying \eqref{eq:NewLemmDebb34} by using \eqref{eq:NewLemmDebb21}-\eqref{eq:NewLemmDebb25}, the expression in \eqref{eq:NewLemmDebb35} is obtained, which is further simplified to the desired result in \eqref{eq:NewLemmDebb36} using the fact that $q_0q_1 = q_2^2$.

Consider the term $\bx{}^{\mr{H}} \bU{}\bV{}\bN{}^{\mr{H}}\bx{} $, which is written as
\begin{IEEEeqnarray}{rCl}
\label{eq:NewLemmDebb41}
&& \bx{}^{\mr{H}} \bU{}\bV{}\bN{}^{\mr{H}}\bx{} \nonumber\\
&& = \bx{}^{\mr{H}} \bU{}\bA{}^{\minus 1}\bN{}^{\mr{H}}\bx{} -q_0 \bx{}^{\mr{H}} \bU{}\bV{}\bx{}\bx{}^{\mr{H}} \bA{}^{\minus 1} \bN{}^{\mr{H}} \bx{} \nonumber\\
&& ~ - q_1\bx{}^{\mr{H}}  \bU{}\bV{}\bw{}\bw{}^{\mr{H}}\bA{}^{\minus 1} \bN{}^{\mr{H}} \bx{}  - q_2\bx{}^{\mr{H}} \bU{}\bV{}\bw{}\bx{}^{\mr{H}}\bA{}^{\minus 1} \bN{}^{\mr{H}} \bx{} \nonumber\\
&& ~ - q_2\bx{}^{\mr{H}} \bU{}\bV{}\bx{}\bw{}^{\mr{H}}\bA{}^{\minus 1} \bN{}^{\mr{H}} \bx{} \\
\label{eq:NewLemmDebb42}
&& = \bx{}^{\mr{H}} \bU{}\bA{}^{\minus 1}\bN{}^{\mr{H}}\bx{} -q_0 \bx{}^{\mr{H}} \bU{}\bV{}\bx{}\bx{}^{\mr{H}}\bA{}^{\minus 1} \bN{}^{\mr{H}} \bx{} \nonumber\\
&& ~ - q_2\bx{}^{\mr{H}} \bU{}\bV{}\bw{}\bx{}^{\mr{H}}\bA{}^{\minus 1} \bN{}^{\mr{H}} \bx{}  \\
\label{eq:NewLemmDebb43}
&& = \parl t_2  -  \frac{q_0  t_1 t_2(1 + q_1 t_1) }{1 + t_1}  + \frac{q_2^2 t_1^2 t_2}{1 + t_1} \parr \frac{ \mr{tr} \curll  \bN{}^{\mr{H}} \curlr }{M}  \\
\label{eq:NewLemmDebb44}
&& = \frac{t_2 (1 + q_1 t_1) }{1 + t_1} \frac{ \mr{tr} \curll  \bN{}^{\mr{H}} \curlr }{M}
\end{IEEEeqnarray}
In \eqref{eq:NewLemmDebb41}, Lemma \ref{lem:ResIdn} is applied, and \eqref{eq:NewLemmDebb42} is reduced to \eqref{eq:NewLemmDebb43} using the result in \eqref{eq:NewLemmDebb25}. Applying Lemma \ref{lem:Lemm7Debb}, along with the results in \eqref{eq:NewLemmDebb22} and \eqref{eq:NewLemmDebb222}, the expression in \eqref{eq:NewLemmDebb43} is obtained. Then, using the fact that $q_0q_1 = q_2^2$, \eqref{eq:NewLemmDebb43} reduces to the desired result in \eqref{eq:NewLemmDebb44}.

Finally, consider the term $\bw{}^{\mr{H}} \bU{}\bV{}\bN{}^{\mr{H}}\bx{}$, which reads as
\begin{IEEEeqnarray}{rCl}
\label{eq:NewLemmDebb51}
&& \bw{}^{\mr{H}} \bU{}\bV{}\bN{}^{\mr{H}}\bx{} \nonumber\\
&& = \bw{}^{\mr{H}} \bU{}\bA{}^{\minus 1}\bN{}^{\mr{H}}\bx{} - q_0 \bw{}^{\mr{H}} \bU{}\bV{}\bx{}\bx{}^{\mr{H}} \bA{}^{\minus 1} \bN{}^{\mr{H}} \bx{} \nonumber\\
&& ~- q_1\bw{}^{\mr{H}}  \bU{}\bV{}\bw{}\bw{}^{\mr{H}}\bA{}^{\minus 1} \bN{}^{\mr{H}} \bx{}  - q_2\bw{}^{\mr{H}} \bU{}\bV{}\bw{}\bx{}^{\mr{H}}\bA{}^{\minus 1} \bN{}^{\mr{H}} \bx{} \nonumber\\
&& ~- q_2\bw{}^{\mr{H}} \bU{}\bV{}\bx{}\bw{}^{\mr{H}}\bA{}^{\minus 1} \bN{}^{\mr{H}} \bx{} \\
\label{eq:NewLemmDebb52}
&& =  -q_0 \bw{}^{\mr{H}} \bU{}\bV{}\bx{}\bx{}^{\mr{H}}\bA{}^{\minus 1} \bN{}^{\mr{H}} \bx{} - q_2\bw{}^{\mr{H}} \bU{}\bV{}\bw{}\bx{}^{\mr{H}}\bA{}^{\minus 1} \bN{}^{\mr{H}} \bx{}  \\
\label{eq:NewLemmDebb53}
&& = \parl \frac{q_0 q_2 t_1^2 t_2  }{1 + t_1}  - \frac{q_2 t_1 t_2 (1 + q_0 t_1)}{1 + t_1} \parr \frac{ \mr{tr} \curll  \bN{}^{\mr{H}} \curlr }{M}  \\
\label{eq:NewLemmDebb54}
&& = - \frac{q_2 t_1 t_2 }{1 + t_1}  \frac{ \mr{tr} \curll  \bN{}^{\mr{H}} \curlr }{M}.
\end{IEEEeqnarray}
The resolvent identity in Lemma \ref{lem:ResIdn} is applied to obtain \eqref{eq:NewLemmDebb51}, which is reduced to \eqref{eq:NewLemmDebb52} using \eqref{eq:NewLemmDebb25}. Furthermore, Lemmas \ref{lem:TrLem}, \ref{lem:Lemm7Debb}, and the result in \eqref{eq:NewLemmDebb22} are used to simplify \eqref{eq:NewLemmDebb53} to  \eqref{eq:NewLemmDebb54}.
\end{proof}
\end{lem}

\section*{Appendix C}  
\section*{Proof of Theorem \ref{thm:DetEqSINR}}
\label{sec:app_c}

The evaluation of the  SINR in \eqref{eq:AchRate5} involves computing the following terms: (i) the normalization constant $\xi$, (ii) the signal term $\temp{sig}$, and (iii) the interference power ${\lVert \boldsymbol{\temp{}}_{\mr{int}} \rVert}^2$. In the sequel, we derive the three terms of interest. Also, when deriving the signal term $\temp{sig}$, we prove Proposition \ref{prop:Irzfsig}.

\subsection{Derivation of $\xi$}

First, we define $\boldsymbol{\chi}_{0} \triangleq \frac{\Hhat{0}^{\mr{H}}\Hhat{0}}{M} +  \alpha \mathbf{I}_M$, and denote by $\boldsymbol{\chi}_{0,\minus k}$ the matrix obtained upon removing the $k$th column of $\boldsymbol{\chi}_{0}$. The RZF precoder $\bG{0}$  in \eqref{eq:SysMod51} satisfies the power constraint $\mr{tr} \parl \bG{0}^{\mr{H}} \bG{0}   \parr  = 1$ implying that
\begin{IEEEeqnarray}{rCl}
\label{eq:XiCnstrnt11}
&& \xi^2 \mr{tr}  \curll \bP{}  \Hhat{0} \parl \Hhat{0}^{\mr{H}}\Hhat{0} + M \alpha \mathbf{I}_M \parr^{\minus 2}   \Hhat{0}^{\mr{H}} \curlr = 1 \\
\label{eq:XiCnstrnt12}
&& \xi^2 \nsum{k=1}{K}  p_k \bhhat{0,k}^{\mr{T}} \parl \Hhat{0,\minus k}^{\mr{H}}\Hhat{0,\minus k} + M \alpha \mathbf{I}_M  +  \bhhat{0,k}^{*} \bhhat{0,k}^{\mr{T}} \parr^{\minus 2}    \bhhat{0,k}^{*} = 1\nonumber\\
\end{IEEEeqnarray}
The expression in \eqref{eq:XiCnstrnt12} is rewritten by applying Lemma \ref{lem:MatInvLem} twice as
\begin{IEEEeqnarray}{rCl}
\label{eq:XiCnstrnt13}
&& \frac{1}{M}\xi^2 \nsum{k=1}{K}   \frac{ \frac{1}{M} p_k  \bhhat{0,k}^{\mr{T}} \boldsymbol{\chi}_{0,\minus k}^{ \minus 2}  \bhhat{0,k}^{*} }{\parl 1 +  \frac{1}{M} \bhhat{0,k}^{\mr{T}}\boldsymbol{\chi}_{0,\minus k}^{\minus 1} \bhhat{0,k}^{*} \parr^2 }     = 1 \\
\label{eq:XiCnstrnt15}
&&  \frac{1}{M} \xi^2  \nsum{k=1}{K}   \frac{p_k  \nsigma{\hhat}{2} \frac{1}{M} \mr{tr} \curll \boldsymbol{\chi}_{0}^{ \minus 2} \curlr }{\parl 1 +   \nsigma{\hhat}{2} \frac{1}{M} \mr{tr} \curll\boldsymbol{\chi}_{0}^{ \minus 1} \curlr  \parr^2 }     = 1 \\
\label{eq:XiCnstrnt16}
&& \xi = \sqrt{ \frac{M}{\nsum{k=1}{K}   \frac{p_k  m'(-\alpha) }{\parl 1 +     m(-\alpha) \parr^2 } } }.
\end{IEEEeqnarray}
Since $ \bhhat{0,k}$ is independent of $ \boldsymbol{\chi}_{0,\minus k}$, and $ \boldsymbol{\chi}_{0,\minus k}$ is non-negative Hermitian, Lemmas \ref{lem:TrLem}, \ref{lem:R1PLem}, and \ref{lem:R1PTrLem} are invoked in order to yield \eqref{eq:XiCnstrnt15} from  \eqref{eq:XiCnstrnt13}. Since the entries of $\Hhat{0}$ are i.i.d. complex Gaussian RVs, Lemma \ref{lem:stiel} is applied to \eqref{eq:XiCnstrnt15}, finally rendering \eqref{eq:XiCnstrnt16} following straightforward algebraic manipulations.

\subsection{Derivation of $\temp{sig}$ }

Define $\bx{0,k} \triangleq  \boldsymbol{\Phi}_{0}^{*}{\bh{k}}^{*}$, where $ \boldsymbol{\Phi}_{0} = \mr{diag} \curll \jexp{\phit{0}{1}}\mathbb{\mathbf{1}}^{\mr{T}}_{1 \times {M}/{M_{\mr{osc}}}}, \ldots  ,\jexp{\phit{0}{M_{\mr{osc}}}} \mathbb{\mathbf{1}}^{\mr{T}}_{1 \times {M}/{M_{\mr{osc}}}} \curlr$. Specifically, $ \boldsymbol{\Phi}_{0} = \mr{diag} \curll \jexp{\phit{0}{1}}, \ldots  ,\jexp{\phit{0}{M}} \curlr$ for the DO setup, and $\boldsymbol{\Phi}_{0} = e^{\jmath \phi_{0}}\mathbf{I}_M$ for the CO setup. Then, the signal component is written as
\begin{IEEEeqnarray}{rCl}
\label{eq:IRzfSig11}
\temp{sig} &&=  \frac{ \sqrt{p_{k} }  \xi }{M}  \bh{k}^{\mr{T}} \bTheta{\tau,k}   \boldsymbol{\chi}_{0}^{\minus 1} \bhhat{0,k}^{*} \\
\label{eq:IRzfSig12}
&&=   \frac{ \sqrt{p_{k}q_0} \xi }{M} \bx{0,k}^{\mr{H}} \Delta \boldsymbol{\Phi}_{\tau} \boldsymbol{\chi}_{0}^{\minus 1} \bx{0,k}    e^{\jmath (\varphit{\tau}{k} - \varphit{0}{k})  } \nonumber\\ && + \frac{ \sqrt{p_{k}q_1}  \xi }{M} \bx{0,k}^{\mr{H}}  \Delta \boldsymbol{\Phi}_{\tau} \boldsymbol{\chi}_{0}^{\minus 1} {\bw{\mr{e},k}}^{*}  e^{\jmath \varphit{\tau}{k}}\\
\label{eq:IRzfSig13}
&&=   \frac{ \sqrt{p_{k}q_0} \xi }{M} \bx{0,k}^{\mr{H}} \Delta \boldsymbol{\Phi}_{\tau} \boldsymbol{\chi}_{0}^{\minus 1} \bx{0,k}     e^{\jmath (\varphit{\tau}{k} - \varphit{0}{k}) }  \nonumber\\ && + \frac{ \sqrt{p_{k}q_1}  \xi }{M} \bx{0,k}^{\mr{H}}  \Delta \boldsymbol{\Phi}_{\tau} \boldsymbol{\chi}_{0}^{\minus 1} {\bw{\mr{e},k}}^{*}  e^{\jmath (\varphit{\tau}{k} - \varphit{0}{k})  }.
\end{IEEEeqnarray}
In \eqref{eq:IRzfSig12}, $\Delta  \boldsymbol{\Phi}_{\tau} = \mr{diag} \curll \right.$ $\jexp{\phit{\tau}{1} - \phit{0}{1}}  \mathbb{\mathbf{1}}^{\mr{T}}_{1 \times {M}/{M_{\mr{osc}}}}$, $\ldots,$ $\left.\jexp{\phit{\tau}{M_{\mr{osc}}} - \phit{0}{M_{\mr{osc}}}}    \mathbb{\mathbf{1}}^{\mr{T}}_{1 \times {M}/{M_{\mr{osc}}}} \curlr$. For the CO setup, $\Delta  \boldsymbol{\Phi}_{\tau} = e^{\jmath(\phi_{\tau} - \phi_{0})} \mathbf{I}_M$, and $ \Delta  \boldsymbol{\Phi}_{\tau} = \mr{diag} \curll \jexp{\phit{\tau}{1} - \phit{0}{1}}  , \ldots  , \jexp{\phit{\tau}{M} - \phit{0}{M}}   \curlr$ for the DO setup.

The expression in \eqref{eq:IRzfSig12} is obtained by using the channel estimate \eqref{eq:SysMod3} in \eqref{eq:IRzfSig11}, and \eqref{eq:IRzfSig12} is further rewritten as \eqref{eq:IRzfSig13} by letting ${\bw{\mr{e},k}} = {\bw{\mr{e},k}} e^{-\jmath \varphit{0}{k}}$, given that ${\bw{\mr{e},k}}$ is circularly symmetric. Define $ t_1 \triangleq \frac{1}{M}   \mr{tr} \curll \boldsymbol{\chi}_{0,\minus k}^{\minus 1} \curlr $. Upon invoking Lemma \ref{lem:FrPr}, for the DO setup we have $\frac{1}{M}   \mr{tr} \curll \Delta \boldsymbol{\Phi}_{\tau} \boldsymbol{\chi}_{0,\minus k}^{\minus 1} \curlr = \frac{1}{M}   \mr{tr} \curll \Delta \boldsymbol{\Phi}_{\tau} \curlr t_1$, where $\Delta \boldsymbol{\Phi}_{\tau} $ is freely independent of $\boldsymbol{\chi}_{0}^{\minus 1}$. This is because $\boldsymbol{\Phi}_{\tau} $ and  $\boldsymbol{\chi}_{0}^{\minus 1}$ contain entries that are statistically independent of each other, and $\boldsymbol{\chi}_{0}^{\minus 1}$  is unitarily invariant \cite{Tao}, \cite{VerduRMT},\cite[Theorem 22.2.3]{Speicher}. Finally, by employing Lemma \ref{lem:Lemm7Debb}, i.e., \eqref{eq:Lemm7Debb}, the signal and noise part in \eqref{eq:IRzfSig13} becomes
\begin{IEEEeqnarray}{rCl} \label{eq:IRzfSig21}
&& \frac{1}{M}\bx{0,k}^{\mr{H}} \Delta \boldsymbol{\Phi}_{\tau} \boldsymbol{\chi}_{0}^{\minus 1} \bx{0,k}  -  \frac{\frac{1}{M}   \mr{tr} \curll \Delta \boldsymbol{\Phi}_{\tau} \curlr t_1(1 + q_1 t_1)}{t_1 + 1}  \overset{M \rightarrow \infty}{\longrightarrow} 0 \IEEEeqnarraynumspace\\
&& \frac{1}{M}\bx{0,k}^{\mr{H}} \Delta \boldsymbol{\Phi}_{\tau} \boldsymbol{\chi}_{0}^{\minus 1} {\bw{\mr{e},k}} -\frac{-q_2 \frac{1}{M}   \mr{tr} \curll \Delta \boldsymbol{\Phi}_{\tau}\curlr  t_1^2}{t_1 + 1}  \overset{M \rightarrow \infty}{\longrightarrow} 0,\IEEEeqnarraynumspace
\end{IEEEeqnarray}

Thus, the desired signal can be written as
\begin{IEEEeqnarray}{rCl}
\label{eq:IRzfSig31}
&&\temp{sig} = \squarel \frac{ { \sqrt{p_{k}q_0}  \xi } \frac{1}{M}   \mr{tr} \curll \Delta \boldsymbol{\Phi}_{\tau}\curlr t_1(1 + q_1 t_1)}{(q_0 + q_1)t_1 + 1}  \right.  \nonumber\\
&& \left. +  { \sqrt{p_{k}q_1}  \xi }   \frac{-q_2 \frac{1}{M}   \mr{tr} \curll \Delta \boldsymbol{\Phi}_{\tau}\curlr t_1^2}{t_1 + 1} \squarer e^{\jmath (\varphit{\tau}{k} - \varphit{0}{k})  }  \IEEEeqnarraynumspace\\
\label{eq:IRzfSig32}
&&=  \frac{\sqrt{q_0{p_{k}}   }\xi   m(-\alpha) \frac{1}{M}   \mr{tr} \curll \Delta \boldsymbol{\Phi}_{\tau}\curlr  }{ m(-\alpha) + 1}     e^{\jmath (\varphit{\tau}{k} - \varphit{0}{k}) },
\end{IEEEeqnarray}
where  Lemma \ref{lem:stiel} is used to reduce \eqref{eq:IRzfSig31} to \eqref{eq:IRzfSig32}, and this corresponds to the result in \eqref{eq:SysMod71} (i.e., Proposition \ref{prop:Irzfsig}).

For the DO setup,  $ \frac{1}{M}   \mr{tr} \curll \Delta \boldsymbol{\Phi}_{\tau}\curlr  = e^{-\frac{\tau \nsigma{\phi}{2}}{2}}$ \cite{Rajet14}, and for the CO setup,  $\frac{1}{M}   \mr{tr} \curll \Delta \boldsymbol{\Phi}_{\tau} \curlr=    e^{ \jmath(\phi_{\tau} - \phi_{0}{})  }$. When $2\leq M_{\mr{osc}} <\infty$,
\begin{IEEEeqnarray}{rCl}
\label{eq:IRzfSig41}
 \frac{1}{M}   \mr{tr} \curll \Delta \boldsymbol{\Phi}_{\tau}\curlr  = \frac{1}{M_{\mr{osc}}} \nnsum{l=1}{M_{\mr{osc}}} \jexp{\phit{\tau}{l} - \phit{0}{l}},
\end{IEEEeqnarray}
and with straightforward calculations, it can be shown that
\begin{IEEEeqnarray}{rCl}
\label{eq:IRzfSig51}
\mathbb{E}_\phi \absl \frac{1}{M}   \mr{tr} \curll \Delta \boldsymbol{\Phi}_{\tau}\curlr \absr ^2 = \frac{  1 - e^{-\tau \nsigma{\phi}{2} } }{M_{\mr{osc}}}  + e^{-\tau \nsigma{\phi}{2} }.
\end{IEEEeqnarray}

\subsection{Derivation of ${\lVert \boldsymbol{\temp{}}_{\mr{int}} \rVert}^2$ }
Define ${\boldsymbol{\bar{\chi}}}_{0} \triangleq {\boldsymbol{\chi}}_{0,\minus k}^{\minus 1}  \Hhat{0,\minus k}^{\mr{H}}   \bP{\minus k}  \Hhat{0,\minus k}  {\boldsymbol{\chi}}_{0}^{\minus 1}$, $ t_2 \triangleq \frac{1}{M}   \mr{tr} \curll {\boldsymbol{\chi}}_{0,\minus k}^{\minus 1}  \Hhat{0,\minus k}^{\mr{H}}   \bP{\minus k}  \Hhat{0,\minus k}  {\boldsymbol{\chi}}_{0}^{\minus 1} \curlr$,  and $\bx{0,k} \triangleq  \boldsymbol{\Phi}_{0}^{*}{\bh{k}}^{*}$. Then, the power of the interference signal $\boldsymbol{\temp{}}_{\mr{int}}$ in \eqref{eq:AchRate5} is evaluated as
\begin{IEEEeqnarray}{rCl}
\label{eq:IRzfInt11}
&& \boldsymbol{\temp{}}_{\mr{int}}^{\mr{H}} \boldsymbol{\temp{}}_{\mr{int}}  \nonumber\\
&&=   \frac{\xi^2}{M^2}   \bh{k}^{\mr{T}} \bTheta{\tau,k}{\boldsymbol{\chi}}_{0}^{\minus 1} \Hhat{0,\minus k}^{\mr{H}} \bP{\minus k}  \Hhat{0,\minus k} {\boldsymbol{\chi}}_{0}^{\minus 1} \bTheta{\tau,k}^{*}{\bh{k}}^{*} \\
\label{eq:IRzfInt12}
&&=  \frac{\xi^2}{M^2}   \bh{k}^{\mr{T}} \boldsymbol{\Phi}_{\tau} {\boldsymbol{\chi}}_{0,\minus k}^{\minus 1} \Hhat{0,\minus k}^{\mr{H}} \bP{\minus k}  \Hhat{0,\minus k}  {\boldsymbol{\chi}}_{0}^{\minus 1} \boldsymbol{\Phi}_{\tau}^{*}{\bh{k}}^{*}\nonumber\\
&& ~+  \frac{\xi^2}{M^2}   \bh{k}^{\mr{T}} \boldsymbol{\Phi}_{\tau} \parl  {\boldsymbol{\chi}}_{0}^{\minus1}   -    {\boldsymbol{\chi}}_{0,\minus k}^{\minus1} \parr \Hhat{0,\minus k}^{\mr{H}} \bP{\minus k}  \Hhat{0,\minus k}   \nonumber\\
&&~~ \cdot  {\boldsymbol{\chi}}_{0}^{\minus 1} \boldsymbol{\Phi}_{\tau}^{*}{\bh{k}}^{*} \\
\label{eq:IRzfInt13}
&&=  \frac{\xi^2}{M^2}   \bh{k}^{\mr{T}}  \boldsymbol{\Phi}_{\tau} {\boldsymbol{\chi}}_{0,\minus k}^{\minus 1} \Hhat{0,\minus k}^{\mr{H}} \bP{\minus k}  \Hhat{0,\minus k}  {\boldsymbol{\chi}}_{0}^{\minus 1} \boldsymbol{\Phi}_{\tau}^{*}{\bh{k}}^{*}\nonumber\\
&&~ -  \frac{\xi^2}{M^2}   \bh{k}^{\mr{T}} \boldsymbol{\Phi}_{\tau} {\boldsymbol{\chi}}_{0}^{\minus 1} \parl  {\boldsymbol{\chi}}_{0}   -    {\boldsymbol{\chi}}_{0,\minus k}\parr  {\boldsymbol{\chi}}_{0,\minus k}^{\minus 1} \Hhat{0,\minus k}^{\mr{H}} \bP{\minus k}  \Hhat{0,\minus k}  \nonumber\\
&&~ \cdot  {\boldsymbol{\chi}}_{0}^{\minus 1} \boldsymbol{\Phi}_{\tau}^{*}{\bh{k}}^{*} \\
\label{eq:IRzfInt15}
&& =  \frac{\xi^2}{M^2}  \bx{\tau,k}^{\mr{H}} \Delta \boldsymbol{\Phi}_{\tau}   {\boldsymbol{\bar{\chi}}}_{0} \Delta \boldsymbol{\Phi}_{\tau} ^{*} \bx{\tau,k}    -  \frac{\xi^2}{M^3}   \bx{\tau,k}^{\mr{H}} \Delta \boldsymbol{\Phi}_{\tau} {\boldsymbol{\chi}}_{0}^{\minus 1}  \nonumber\\
&&~~ \cdot \parl  q_0 \bx{\tau,k}\bx{\tau,k}^{\mr{H}}  +  q_1 \bwu{\tau,k}^{*}\bwu{\tau,k}^{\mr{T}}  +  q_2\bx{\tau,k} \bwu{\tau,k}^{\mr{T}} \right. \nonumber\\
&&~~~  \left.   +  q_2 \bwu{\tau,k}^{*} \bx{\tau,k}^{\mr{H}} \parr {\boldsymbol{\bar{\chi}}}_{0} \Delta \boldsymbol{\Phi}_{\tau}^{*}  \bx{\tau,k}.
\end{IEEEeqnarray}
The resolvent identity lemma (i.e., Lemma \ref{lem:ResIdn}) is applied in \eqref{eq:IRzfInt12} resulting in the expression in \eqref{eq:IRzfInt13}. The definitions  of $ {\boldsymbol{\chi}}_{0}$, and $ {\boldsymbol{\chi}}_{0,\minus k}$, and the channel estimate in \eqref{eq:SysMod3} are employed in \eqref{eq:IRzfInt13} to obtain \eqref{eq:IRzfInt15}. Finally, the quadratic forms in \eqref{eq:IRzfInt15} are simplified using Lemmas \ref{lem:Lemm7Debb} and \ref{lem:NewLemDebb} as
\begin{IEEEeqnarray}{rCl}
\label{eq:IRzfInt21}
&& \frac{\bx{\tau,k}^{\mr{H}}  \Delta \boldsymbol{\Phi}_{\tau}  {\boldsymbol{\bar{\chi}}}_{0} \Delta \boldsymbol{\Phi}_{\tau} ^{*} \bx{\tau,k} }{M^2}  -  t_2 + \frac{q_0 t_1 t_2 \absl \frac{ \mr{tr} \curll  \Delta\boldsymbol{\Phi}_{\tau} \curlr }{M} \absr^2}{1 + t_1}   \underset{\eqref{eq:NewLemmDebb01}}{ \overset{M \rightarrow \infty}{\longrightarrow} } 0 \IEEEeqnarraynumspace\\
\label{eq:IRzfInt22}
&&\frac{\bx{\tau,k}^{\mr{H}} \Delta \boldsymbol{\Phi}_{\tau}  {\boldsymbol{\chi}}_{0}^{\minus 1} \bx{\tau,k} }{M^2}  -\frac{ \frac{ \mr{tr} \curll  \Delta\boldsymbol{\Phi}_{\tau} \curlr }{M} t_1(1 + q_1 t_1)}{1 + t_1}   \underset{\eqref{eq:Lemm7Debb}}{ \overset{M \rightarrow \infty}{\longrightarrow}} 0 \IEEEeqnarraynumspace\\
\label{eq:IRzfInt23}
&& \frac{ \bx{\tau,k}^{\mr{H}} \Delta \boldsymbol{\Phi}_{\tau}  {\boldsymbol{\chi}}_{0}^{\minus 1} \bwu{\tau,k}^{*} }{M^2} -  \frac{ -q_2 t_1^2 }{1 + t_1} \underset{\eqref{eq:Lemm7Debb}}{ \overset{M \rightarrow \infty}{\longrightarrow}}  0 \IEEEeqnarraynumspace\\
\label{eq:IRzfInt24}
&& \frac{ \bx{\tau,k}^{\mr{H}} {\boldsymbol{\bar{\chi}}}_{0} \Delta \boldsymbol{\Phi}_{\tau} ^{*}\bx{\tau,k} }{M^2} - \frac{ \frac{ \mr{tr} \curll  \Delta\boldsymbol{\Phi}_{\tau}^{*} \curlr }{M}t_2 (1 + q_1 t_1) }{1 + t_1}   \underset{\eqref{eq:NewLemmDebb02}}{ \overset{M \rightarrow \infty}{\longrightarrow} } 0 \\
\label{eq:IRzfInt25}
&& \frac{ \bwu{\tau,k}^{\mr{T}}{\boldsymbol{\bar{\chi}}}_{0} \Delta \boldsymbol{\Phi}_{\tau} ^{*} \bx{\tau,k}}{M^2} - \frac{-q_2 t_1 t_2 \frac{ \mr{tr} \curll  \Delta\boldsymbol{\Phi}_{\tau}^{*} \curlr }{M}}{1 + t_1}  \underset{\eqref{eq:NewLemmDebb03}}{ \overset{M \rightarrow \infty}{\longrightarrow} } 0.
\end{IEEEeqnarray}

In order to evaluate $t_2$, we invoke Lemmas \ref{lem:R1PLem} and \ref{lem:R1PTrLem} to obtain
\begin{IEEEeqnarray}{rCl}
\label{eq:IRzfInt31}
\frac{\mr{tr} \curll  \bP{\minus k}  \Hhat{0,\minus k} \parl  {\boldsymbol{\chi}}_{0,\minus k}^{-2} -  {\boldsymbol{\chi}}_{0}^{-2} \parr \Hhat{0,\minus k}^{\mr{H}}       \curlr  }{M }   \overset{M \rightarrow \infty}{\longrightarrow} 0.
\end{IEEEeqnarray}
Thus, for $M \rightarrow \infty$, we simplify $ t_2 \approx $$\frac{1}{M}   \mr{tr}  \curll   \bP{\minus k}    \Hhat{0,\minus k}  {\boldsymbol{\chi}}_{0}^{\minus 2}  \Hhat{0,\minus k}^{\mr{H}}  \curlr$ as
\begin{IEEEeqnarray}{rCl}
\label{eq:IRzfInt40}
{t_2} &=& \nnsum{\stackrel{k_1=1,}{k_1 \neq k}}{K} \frac{1}{M}  p_{k_1}  \bhhat{0,k_1}^{\mr{T}}  {\boldsymbol{\chi}}_{0 }^{-2}  \bhhat{0,k_1}^{*}     \\
\label{eq:IRzfInt41}
&=& \nnsum{\stackrel{k_1=1,}{k_1 \neq k}}{K} \frac{1}{M}  p_{k_1}  \bhhat{0,k_1}^{\mr{T}} \parl  {\boldsymbol{\chi}}_{0,\minus k_1} + \bhhat{0,k_1}^{*}\bhhat{0,k_1}^{\mr{T}}      \parr^{-2}  \bhhat{0,k_1}^{*}     \\
\label{eq:IRzfInt42}
&=&  \nnsum{\stackrel{k_1=1,}{k_1 \neq k}}{K}   p_{k_1}   \frac{\frac{1}{M}\bhhat{0,k_1}^{\mr{T}}   {\boldsymbol{\chi}}_{0,\minus k_1}^{-2}  \bhhat{0,k_1}^{*} }{( 1 + \frac{1}{M} \bhhat{0,k_1}^{\mr{T}}  {\boldsymbol{\chi}}_{0,\minus k_1}^{\minus 1} \bhhat{0,k_1}^{*} )^2}    \\
\label{eq:IRzfInt44}
&=&  \nnsum{\stackrel{k_1=1,}{k_1 \neq k}}{K}   p_{k_1}   \frac{ \nsigma{\hhat}{2} \frac{1}{M} \mr{tr}  \curll {\boldsymbol{\chi}}_{0 }^{-2} \curlr  }{( 1 + \nsigma{\hhat}{2} \frac{1}{M}  \mr{tr}  \curll {\boldsymbol{\chi}}_{0 }^{\minus 1} \curlr   )^2}    \\
\label{eq:IRzfInt45}
&=& \nnsum{\stackrel{k_1=1,}{k_1 \neq k}}{K}   p_{k_1}   \frac{  m'(-\alpha)  }{( 1 +   m(-\alpha)    )^2}
\end{IEEEeqnarray}
In order to obtain \eqref{eq:IRzfInt42}, Lemma \ref{lem:MatInvLem} is applied twice to \eqref{eq:IRzfInt41}. Further, Lemmas \ref{lem:TrLem}, \ref{lem:R1PLem}, and \ref{lem:R1PTrLem}  are employed in \eqref{eq:IRzfInt42} to obtain \eqref{eq:IRzfInt44}, and Lemma \ref{lem:stiel} is applied to  \eqref{eq:IRzfInt44} yielding  \eqref{eq:IRzfInt45}. Hence, the interference power can be reduced as
\begin{IEEEeqnarray}{rCl} \label{eq:IRzfInt51}
&& \boldsymbol{\temp{}}_{\mr{int}}^{\mr{H}} \boldsymbol{\temp{}}_{\mr{int}} \nonumber\\ &&=  \frac{\xi^2}{M} \parl t_2 -  \frac{q_0 t_1 t_2 }{1 + t_1} \absl \frac{ \mr{tr} \curll   \Delta\boldsymbol{\Phi}_{\tau}  \curlr }{M}\absr^2  \parr  -  \frac{\xi^2}{M}  \nonumber\\
&& \cdot \parl  \absl \frac{ \mr{tr} \curll  \Delta\boldsymbol{\Phi}_{\tau} \curlr }{M} \absr^2  \frac{t_1 t_2 (1 + q_1 t_1)^2 }{(1 + t_1)^2}  +  \absl \frac{ \mr{tr} \curll  \Delta\boldsymbol{\Phi}_{\tau} \curlr }{M}\absr^2     \frac{q_2^2 t_1^3 t_2 }{(1 + t_1)^2}   \right. \nonumber\\
&& \left. - 2 \absl \frac{ \mr{tr} \curll  \Delta\boldsymbol{\Phi}_{\tau} \curlr }{M} \absr^2  \frac{q_2 t_1^2 t_2(1 + q_1 t_1)}{(1 + t_1)^2}     \parr \\
\label{eq:IRzfInt52}
&&=  \frac{\xi^2}{M} \parl t_2 -  \frac{q_0 t_1 t_2  \absl \frac{ \mr{tr} \curll  \Delta\boldsymbol{\Phi}_{\tau} \curlr }{M}\absr^2}{1 + t_1}   -  \frac{ q_0 t_1 t_2 \absl \frac{ \mr{tr} \curll  \Delta\boldsymbol{\Phi}_{\tau} \curlr }{M} \absr^2}{(1 + t_1)^2}   \parr,
\end{IEEEeqnarray}
where $t_2$ is given in \eqref{eq:IRzfInt45}, and $t_1 = m(-\alpha)$ for $M\rightarrow \infty$. 

\bibliographystyle{IEEEbib}

\end{document}